\documentclass{article}
\usepackage[utf8]{inputenc}
\usepackage{amsmath}
\usepackage{amsfonts}
\usepackage{amssymb}
\usepackage{graphicx}
\usepackage[title,titletoc]{appendix}
\usepackage{fullpage}
\usepackage{color}
\usepackage[sort&compress,numbers]{natbib}
\usepackage{caption}
\usepackage{subcaption}
\usepackage{float}
\usepackage[normalem]{ulem} 
\usepackage{multirow}
\usepackage{caption}
\usepackage{bm}
\usepackage{array}

\usepackage{amsthm}
\newtheorem{theorem}{Theorem}[section]
\newtheorem{lemma}[theorem]{Lemma}
\newtheorem{corollary}[theorem]{Corollary}
\newtheorem{definition}{Definition}

\numberwithin{subcase}{case}

\newtheorem{remark}[theorem]{Remark}

\usepackage{algorithm} 
\usepackage{algpseudocode} 

\usepackage{xcolor} 

\begin{document}
\algnewcommand{\algorithmicgoto}{\textbf{go to}}%
\algnewcommand{\Goto}[1]{\algorithmicgoto~\ref{#1}}%
\title{\normalsize\bf \vspace{-8ex}
Online landmark replacement for out-of-sample dimensionality reduction methods}
\author{Chanon Thongprayoon$^1$, Naoki Masuda$^{1,2,3,*}$}
\date{\vspace{-5ex}}
\maketitle
\vspace{-14pt}
\begin{center}
\small
$^1$Department of Mathematics, State University of New York at Buffalo, Buffalo, NY, USA.\\
$^2$Institute for Artificial Intelligence and Data Science, State University of New York at Buffalo, Buffalo, NY, USA\\
$^3$Center for Computational Social Science, Kobe University, Kobe, Japan\\
*Corresponding author: naokimas@buffalo.edu
\end{center}

\begin{abstract}
A strategy to assist visualization and analysis of large and complex data sets is dimensionality reduction, with which one maps each data point into a low-dimensional manifold. However, various dimensionality reduction techniques are computationally infeasible for large data. Out-of-sample techniques aim to resolve this difficulty; they only apply the dimensionality reduction technique on a small portion of data, referred to as landmarks, and determine the embedding coordinates of the other points using landmarks as references. Out-of-sample techniques have been applied to online settings, or when data arrive as time series. However, existing online out-of-sample techniques use either all the previous data points as landmarks or the fixed set of landmarks and therefore are potentially not good at capturing the geometry of the entire data set when the time series is non-stationary. To address this problem, we propose an online landmark replacement algorithm for out-of-sample techniques using geometric graphs and the minimal dominating set on them. We mathematically analyze some properties of the proposed algorithm, particularly focusing on the case of landmark multidimensional scaling as the out-of-sample technique, and test its performance on synthetic and empirical time series data.
\end{abstract}

\section{Introduction\label{sec:introduction}}

One way to aid visualization and analysis of high-dimensional data is to seek representations of the given set of data points in low-dimensional spaces, including the case of non-Euclidean manifolds \cite{duin2005dissimilarity, Desilva2003Nips, ayesha2020overview}. Finding a low-dimensional representation of high-dimensional data is referred to as dimensionality reduction \cite{van2009dimensionality, lee2007nonlinear, cunningham2015linear}. Examples of dimensionality reduction methods are the principal component analysis (PCA) \cite{abdi2010principal}, multidimensional scaling (MDS) \cite{borg2005modern}, locally linear embedding (LLE) \cite{roweis2000lle}, isometric feature mapping (ISOMAP) \cite{tenenbaum2000global}, and $t$-distributed stochastic neighbor embedding ($t$-SNE) \cite{van2008visualizing}. These and many other methods aim to preserve the dissimilarity among data points as much as possible through the dimensionality reduction process \cite{van2009dimensionality, gisbrecht2015data}. Applications of dimensionality reduction include pattern recognition, statistical analysis of multi-variable data, complexity reduction, feature extraction, and signal processing, to name a few \cite{duin2005dissimilarity, carreira1997review, velliangiri2019review, ghojogh2023elements}.

In general, running a dimensionality reduction algorithm may require a prohibitive amount of time when a data set is large. For example, the PCA, MDS, and ISOMAP have the cubic time complexity \cite{van2009dimensionality}, while other methods such as LLE and $t$-SNE have the quadratic time complexity~\cite{van2009dimensionality, van2008visualizing}. One strategy for suppressing the time complexity of dimensionality reduction algorithms is to approximate the embedding coordinates of the data points. To do so, one can turn to approximate minimization of the cost function of an embedding method \cite{schleif2015indefinite, van2013barnes, yang2013scalable}. Alternatively, one can construct a computationally feasible projection matrix, whose entries are random variables that take values from probability distributions of data points of one's choice, such as the Gaussian distribution \cite{dasgupta2013experiments, lopes2011more, durrant2010compressed} or the distributions introduced in \cite{achlioptas2003database, bingham2001random}. Such a projection matrix applied to the data points in the original high-dimensional Euclidean space provides their approximate embedding coordinates in a low-dimensional Euclidean space \cite{achlioptas2003database, schleif2015indefinite, durrant2010compressed, bingham2001random, dasgupta2013experiments, lopes2011more}.  Another strategy is to exactly embed only a relatively small subset of the entire data, referred to as out-of-sample techniques \cite{GISBRECHT2015kerneltsne, Bengio2004Nips, zhang2021bikerneltsne, long2019landmark, gisbrecht2015metric, van2009dimensionality}. An out-of-sample dimensionality reduction method calculates the precise embedding coordinate only for a small portion of data points, called the landmarks, and approximates the coordinates for the remaining data points. Because one only approximately calculates the coordinates of the non-landmark data points, the embedding of the entire data set with an out-of-sample method usually runs faster than the original embedding method. Examples of out-of-sample techniques are the landmark multidimensional scaling (LMDS) \cite{Desilva2003Nips, de2004sparse}, landmark isometric feature mapping (L-ISOMAP) \cite{Desilva2003Nips}, landmark diffusion maps \cite{long2019landmark}, kernel $t$-SNE \cite{GISBRECHT2015kerneltsne}, and bi-kernel $t$-SNE \cite{zhang2021bikerneltsne}.
 
We have discussed embeddings of static data sets. In practice, one may want to embed a stream of data arriving in real time into a low-dimensional space. Out-of-sample techniques are directly applicable to such online settings because one can calculate the embedding coordinates of newly arriving data points, which are regarded as out-of-sample data. In general, online methods should suppress the computational time by avoiding repeatedly applying a dimensionality reduction procedure every time a new data point arrives \cite{migenda2021adaptive, hoi2021online}. Known online dimensionality reduction techniques include online PCA \cite{migenda2021adaptive, hall1998incremental, artac2002incremental, cardot2018online, boutsidis2014online}, online Laplacian eigenmaps \cite{malik2016online}, online ISOMAP \cite{law2004nonlinear, law2006incremental}, online L-ISOMAP \cite{law2006incremental}, and online LLE \cite{schuon2008truly}. There are also downstream tasks facilitated by online dimensionality reduction techniques for stream data, such as online anomaly detection for cybersecurity \cite{juvonen2015online}, online classification and clustering of data points \cite{zhou2016online, wang2013online, liberty2016algorithm, hoi2014libol}, and autonomous control in robotics \cite{gepperth2016incremental}.

The online dimensionality reduction techniques mentioned above are instances of out-of-sample techniques. Regarding the selection of landmarks, there are two major classes of these methods. First, some algorithms use all the previous data points as landmarks and therefore the number of landmarks keeps growing \cite{hall1998incremental, artac2002incremental, cardot2018online, boutsidis2014online, migenda2021adaptive, law2004nonlinear, law2006incremental, malik2016online, schuon2008truly}. This class of methods does not restrict the number of landmarks and hence does not particularly aim to reduce runtime. Second, other algorithms use a set of landmarks that is fixed over time once it is initialized \cite{de2004sparse, GISBRECHT2015kerneltsne, long2019landmark, zhang2021bikerneltsne, law2006incremental}. An earlier study suggested that using evenly distributed landmarks preserves the data structure better than using a set of poorly distributed landmarks \cite{de2004sparse}. Therefore, one idea is to choose landmarks from the data set uniformly at random. However, the new data points may explore an unexplored part of the space of the data, and they may be far away from any of the already revealed data points. In this case, 
embedding with an initially fixed set of landmarks may cause substantial distortion because they are unevenly distributed given the data set including the newer data points. It may then be better to adaptively replace some landmarks as new data points arrive to assure that the majority of landmarks do not congregate in a small region of the data space and the set of landmarks provides a reasonable scaffold of the entire data set at any given time. We currently lack such landmark replacement algorithms for online out-of-sample embedding.

We remark that there are online dimensionality reduction techniques that do not use out-of-sample techniques. In fact, they require one to increase the dimension of the embedding space as new data points arrive \cite{migenda2021adaptive, hall1998incremental, artac2002incremental, boutsidis2014online}, which may be inconvenient in practice.

To fill this gap, we study online landmark selection for out-of-sample techniques. By constructing an adaptive geometric graph \cite{penrose2003random, barthelemy2011spatial} from the given data set and exploiting an algorithm that finds the online dominating set \cite{hjuler2019dominating}, we propose an algorithm of online landmark replacement for out-of-sample dimensionality reduction methods. Upon the arrival of a new data point, the proposed algorithm determines whether or not a current landmark should be replaced. We then provide mathematical underpinnings of the proposed algorithm, specifically in the case of the LMDS, and apply the proposed algorithm to three multidimensional time series data sets.

\section{Online landmark replacement algorithm}

In this section, we propose an online landmark replacement algorithm for out-of-sample embedding methods.
We introduce building blocks for our algorithm in sections~\ref{sub:geometric-graphs}, \ref{sub:offline_landmark}, and \ref{sub:online_landmark}.
Then, we explain our online algorithm in section~\ref{sub:proposed_online_landmark}.

\subsection{Geometric graph, dominating set, and landmark\label{sub:geometric-graphs}}

\begin{figure}[t]
    \centering
     \includegraphics[width=0.6 \textwidth]{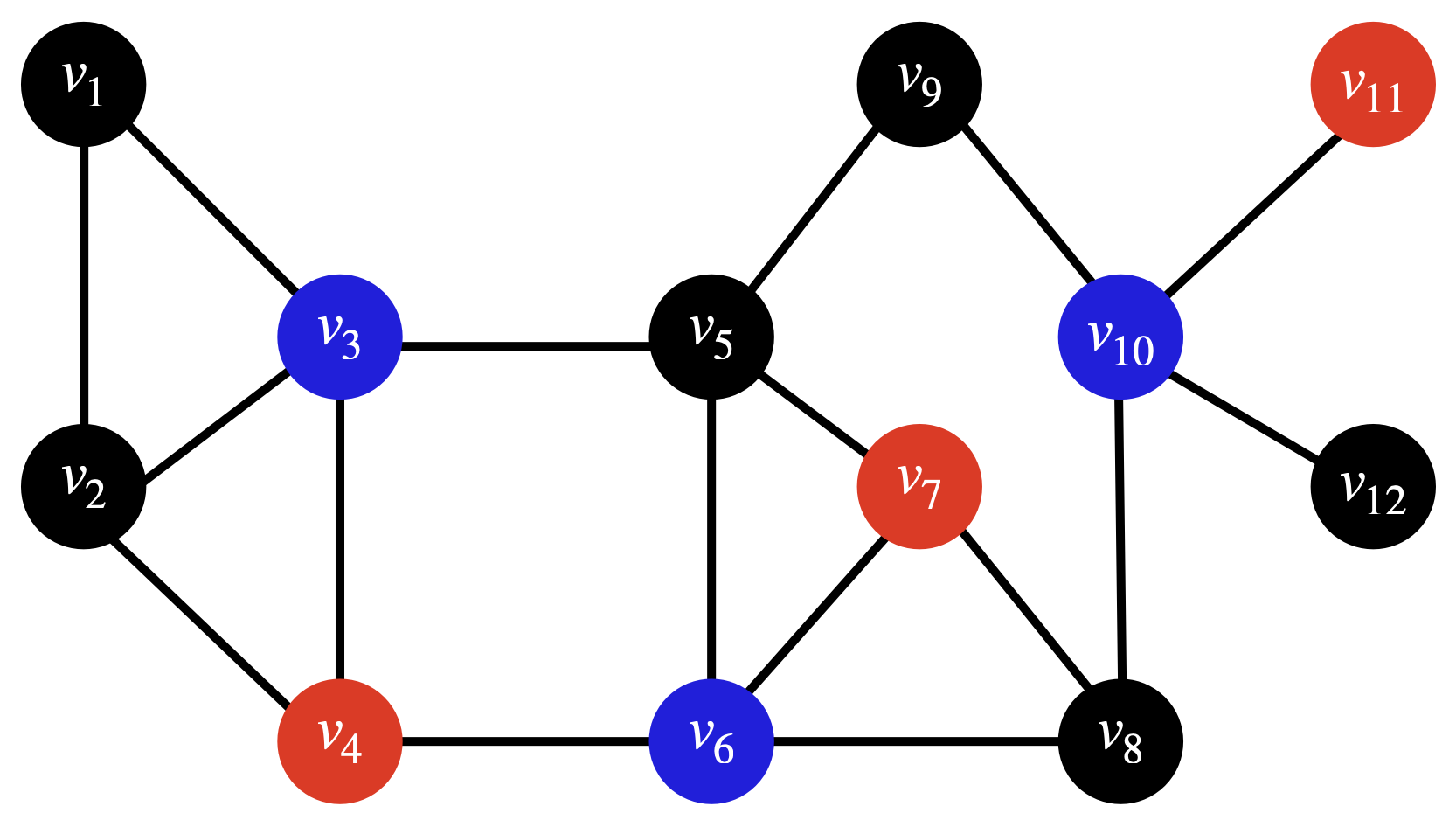}
     \caption{A graph with 12 nodes. The set $L = \{v_3, v_6, v_{10}\}$, shown in blue, is a dominating set because, for each $r\in\{1,\ldots, 12\}$, $v_r$ either belongs to $L$ or is adjacent to an element in $L$. On the other hand, $L_0 = \{v_4, v_7, v_{11}\}$, shown in red, is not a dominating set because $v_1, v_9,$ and $v_{12}$ do not belong to $L_0$ or is adjacent to an element in $L_0$. Here, $OnlyBy(v_3)=\{v_1,v_2,v_3,v_4\}$. For instance, $v_5\not\in OnlyBy(v_3)$ because $v_5$ is dominated by both $v_3$ and $v_6$.}
\label{fig:dominating_set_onlyby}
\end{figure}

In this section, we introduce the notion of landmark via the definition of geometric graph and dominating set.
We let $V = \{x_1,x_2,\ldots, x_n\}$ be the set of data, arriving at $0\le t_1 < t_2 < \cdots < t_n$, $d:V\times V\rightarrow \mathbb{R}$ be a metric, and $\rho$ be a positive number. We define a network with node set $V$ and edge set
given by $E = \{(x_r, x_s): d(x_r, x_s) < \rho\}$. This network is called a geometric graph \cite{penrose2003random, barthelemy2011spatial}. We also denote by $L$ $(\subset V)$ the set of landmarks. Our problem of landmark replacement is to update $L$ that contains at most $m$ landmarks given $V$ and $d(x_r, x_s)$ $\forall r, s \in \{1, \ldots, n \}$ such that $L$ has some desirable properties. We fix $m$.

To find $L$, we turn to the dominating set problem \cite{rahman2017basic, hedetniemi1991bibliography}. The set $L\subseteq V$ is said to be a dominating set of a graph $(V,E)$ if, for each node $v$, either $v$ belongs to $L$ or $v$ is adjacent to a node in $L$. In the former case, we say that $v$ covers itself (see Fig.~\ref{fig:dominating_set_onlyby} for an illustration). In the latter case, we say that $v$ is covered by a different node in $L$. If $v$ neither belongs to $L$ nor is covered by a different node in $L$, then we say that $v$ is uncovered. If there is at least one $v$ that is uncovered, then $L$ is not a dominating set. We aim to select a dominating set $L$ to be used as the set of landmarks for out-of-sample embedding methods.

\subsection{Selecting landmarks offline\label{sub:offline_landmark}} 

Here, we introduce a procedure to determine $L$ in an offline setting without an upper limit on the number of landmarks. We assume that $V$ is given and fixed.
If $n\le m$, then we set $L = V$. Otherwise, we apply Algorithm \ref{alg:greedy_ds} to find a dominating set $L\subseteq V$ for network $(V,E)$ \cite{parekh1991analysis} and use $L$ as the set of landmarks. A dominating set, obtained from Algorithm \ref{alg:greedy_ds}, attains the size of $O(\log(n))\times \left|OPT\right|$, where $\left|OPT\right|$ is the size of the optimal dominating set~\cite{hjuler2019dominating}; see electronic supplementary material section~\ref{sec:bigO} for the formal definition of the Big $O$ notation.

\begin{algorithm}
	\caption{Greedy dominating set} 
	\label{alg:greedy_ds}
	\begin{algorithmic}[1]
	\State Set $L = \emptyset$
		\While {there is an uncovered node}
			\State Choose the $x_r \notin L$ such that $x_r$ covers the largest number of uncovered nodes.
			\State $L \gets L\cup \{x_r\}$
		\EndWhile
	\State \Return $L$
	\end{algorithmic} 
\end{algorithm}

\subsection{Previous algorithms for updating landmarks online\label{sub:online_landmark}}

We now turn our interest to the case of online dominating set. There are algorithms that allow updating of the dominating set as nodes are sequentially added or deleted~\cite{eidenbenz2002online, bockenhauer2021advice, king1997line, hjuler2019dominating}. For instance, one method recomputes a dominating set on time-varying unit disk graphs upon each addition of deletion of a node~\cite{eidenbenz2002online}. In this algorithm, there is no limitation to the size of the dominating set, which we denote by $\left| L \right|$. If the added node is covered by the current dominating set, then no action is required. Otherwise, we include the added node into the dominating set. When deleting a node $v$ from the network, if $v$ does not belong to the dominating set, no action is taken. Otherwise, there may be uncovered nodes as a result of deleting $v$ from the dominating set. If the latter is the case, we add all the uncovered neighbors of $v$ to the dominating set. This last step may drastically increase $\left| L \right|$. Another online algorithm, which is the latter of the two algorithms proposed in \cite{hjuler2019dominating}, seeks to find a minimal dominating set and does not restrict $m$. In addition to allowing sequentially adding nodes to $L$, this algorithm permits sequential addition of edges to the network, which is a useful property for our landmark replacemenet algorithm as we explain later. Two other online algorithms do not allow us to add an edge between the existing nodes \cite{king1997line}. Reference~\cite{bockenhauer2021advice} provides lower and upper bounds on the size of dominating sets for deterministic and randomized online dominating set algorithms. 

All these algorithms allow $\left| L \right|$ to indefinitely grow and do not provide a mechanism to control $\left| L \right|$. In contrast, we want to maintain the dominating set to satisfy $\left| L \right| \le m$ as we add nodes. Therefore, we propose to allow $\rho$ to adaptively increase, which is equivalent to adding edges to the network, with the aim of decreasing $\left| L \right|$ when $\left| L \right|$ exceeds $m$. 
Although the two algorithms proposed in \cite{king1997line}, discussed above, allow us to add edges, one can only add new edges incident to any new node $v$ when $v$ has joined the network. These algorithms do not allow us to add an edge $(v, u)$ at a later stage, which we may want to do in our landmark replacement procedure in which we increase $\rho$ as we add nodes, as we describe below. Therefore, we need to resort to a different online dominating set algorithm to realize our goal.

\subsection{Proposed algorithm for updating landmarks online\label{sub:proposed_online_landmark}}

In this section, we explain our algorithm to update the set of landmarks, $L$, upon the arrival of each data point under the assumption that the data points sequentially arrive. We limit the size of $L$ by imposing $\left| L\right| \le m$.
Whenever a new data point arrives after the first $m$ data points, we either keep the current $L$ or replace an element in $L$ by another element to respect $\left| L\right| \le m$.

We adapt an online algorithm that finds the minimal dominating set of time-varying networks \cite{hjuler2019dominating} to our landmark replacement problem. With our proposed method, we control the number of landmarks in a dynamic manner by gradually increasing $\rho$, which corresponds to adding edges to the geometric graph.

We assume that $V = \{ x_1, \ldots, x_n \}$ and that $x_{n+1}$ is a new data point. We regard that $x_{n+1}$ is a new node of the network, which directly connects to each of the existing nodes within distance $\rho$ from $x_{n+1}$ by a new edge. Then, exactly one of the following three statements holds true:
\begin{enumerate}
\item $x_{n+1}$ is not isolated and is covered by $L$. In other words, $\mathcal{N}(x_{n+1})\cup\{x_{n+1}\}\neq \{x_{n+1}\}$ and $\left(\mathcal{N}(x_{n+1})\cup\{x_{n+1}\}\right)\cap L \neq\emptyset$,
\item $x_{n+1}$ is not isolated and is not covered by $L$. In other words, $\mathcal{N}(x_{n+1})\cup\{x_{n+1}\}\neq \{x_{n+1}\}$ and $\left(\mathcal{N}(x_{n+1})\cup\{x_{n+1}\}\right)\cap L =\emptyset$,
\item $x_{n+1}$ is isolated. In other words, $\mathcal{N}(x_{n+1})\cup\{x_{n+1}\} = \{x_{n+1}\}$,
\end{enumerate}
where $\mathcal{N}(x_{n+1})$ is the set of neighbors of $x_{n+1}$. 

In the first case, $L$ continues to be a dominating set of the graph. Therefore, we do not need to update $L$ given the current value of $\rho$. Therefore, there is no need to update $\rho$. In the second and third cases, we update $L$. To discuss the latter two cases, we define $OnlyBy(x)$ for any $x\in L$ as follows:
\begin{equation}
OnlyBy(x) = \{y\in\mathcal{N}(x)\cup\{x\}: \left| \left(\mathcal{N}(y)\cup\{y\}\right)\cap L\right| = 1\}\label{eq:unique_cover}.
\end{equation}
In other words, $OnlyBy(x)$ is the set of neighbors of $x$ that only $x$ dominates (see Fig.~\ref{fig:dominating_set_onlyby} for an illustration).

Let us consider the second case. If $\left| L\right| < m$, then we add $x_{n+1}$ to the set of landmarks, $L$. 
The addition of $x_{n+1}$ to $L$ does not violate $\left| L\right| \le m$. Otherwise, i.e., if $\left| L\right| = m$ before the arrival of $x_{n+1}$, we use Algorithm~\ref{alg:landmark_replacement} to determine if we replace a landmark in $L$ by $x_{n+1}$ or by another non-landmark data point.  Algorithm~\ref{alg:landmark_replacement} either adds $x_{n+1}$ or a neighbor of $x_{n+1}$ to $L$; we make the decision based on $\deg(x_{n+1})$, which is the degree of node $x_{n+1}$. Then, because $\left| L\right|$ has now increased to $m+1$, we need to drop one data point from $L$ to enforce $\left| L\right| = m$. We realize this by increasing $\rho$, which leads to addition of new edges to the graph, until one landmark becomes redundant. Specifically, we add an edge between a pair of nodes that attains the smallest distance among all the node pairs with distance greater than the current $\rho$. This entails an increase in $\rho$. If neither of the two nodes joined by the new edge is an element of $L$, then we further increase $\rho$ such that the next pair of nodes is additionally joined by an edge. We repeat adding edges in this manner until either of the two nodes that are joined by the new edge is a landmark. Once the repetitive addition of edges has been completed, there may be a redundant landmark that can be removed. Therefore, to decide which landmark will be replaced, we call Algorithm~\ref{alg:landmark_removal_condition} (i.e., function REMOVE) \cite{hjuler2019dominating}.
Algorithm~\ref{alg:landmark_removal_condition} returns a landmark after the removal of which $L$ remains a dominating set. If Algorithm~\ref{alg:landmark_removal_condition} returns such a landmark, then Algorithm~\ref{alg:landmark_replacement} removes this landmark. Otherwise, i.e., if Algorithm~\ref{alg:landmark_removal_condition} returns an empty set, Algorithm~\ref{alg:landmark_replacement} sequentially keeps adding edges by gradually increasing $\rho$ until REMOVE returns a landmark that can be removed (which Algorithm~\ref{alg:landmark_replacement} then removes). Lines \ref{code:our_contribution_1}-\ref{code:our_contribution_2} and \ref{code:our_contribution_3}-\ref{code:our_contribution_4} in Algorithm~\ref{alg:landmark_replacement} are novel, and the remainder of Algorithm~\ref{alg:landmark_replacement} including Algorithm~\ref{alg:landmark_removal_condition} is the same as the original one \cite{hjuler2019dominating}.

In the third case, $x_{n+1}$ is an isolated node in $(V, E)$. This is similar to the second case in that $x_{n+1}$ is not adjacent to any element in $L$. If $\left| L\right| < m$, then we just add $x_{n+1}$ to $L$. Note that the addition of $x_{n+1}$ respects  $\left| L\right| \le m$. Otherwise, i.e., if $\left| L \right| = m$ before the arrival of $x_{n+1}$, we apply Algorithm~\ref{alg:landmark_replacement} to $(V,E)$. Because $\deg(x_{n+1}) = 0$, lines~\ref{code:smaller_degree_x_n+1} and \ref{code:include_x_n+1} of Algorithm~\ref{alg:landmark_replacement} imply that one adds $x_{n+1}$ to $L$. Then, there are now $m+1$ landmarks, and we need to remove one of them. Therefore, we carry out lines \ref{code:our_contribution_1}-\ref{code:our_contribution_4} in Algorithm~\ref{alg:landmark_replacement} to remove a landmark, recovering $\left| L \right| = m$.

\begin{algorithm}
\caption{Landmark replacement}
\label{alg:landmark_replacement}
	\hspace*{\algorithmicindent} \textbf{Input} $(V,E)$
	\begin{algorithmic}[1]
	\If {$\deg(x_{n+1})\le 2\sqrt{\left| E\right|}$}\label{code:smaller_degree_x_n+1}
		\State $L\gets L\cup\{x_{n+1}\}$\label{code:include_x_n+1}
	\Else
		\State Choose an arbitrary $\tilde{x}\in\mathcal{N}(x_{n+1})$ such that $\deg(\tilde{x})\le\sqrt{\left|E\right|}$\label{code:search_neighbor_x_n+1}
		\State $L\gets L\cup\{\tilde{x}\}$\label{code:include_x_tilde}
	\EndIf
	\State $\left(x_r, x_s\right) = \text{argmin}_{\{x_{r'}, x_{s'}\}\not\in E}d(x_{r'}, x_{s'})$\label{code:our_contribution_1}
	\State Update $\rho = d(x_r, x_s)$ and $E\gets E\cup \{(x_r,x_s)\}$\label{code:first_epsilon_update}
	\While {$x_r\not\in L \text{ and } x_s\not\in L$}\label{code:start_first_while_loop}
		\State $\left(x_r, x_s\right) = \text{argmin}_{\{x_{r'}, x_{s'}\}\not\in E}d(x_{r'}, x_{s'})$
		\State Update $\rho = d(x_r, x_s)$ and $E\gets E\cup \{(x_r,x_s)\}$
	\EndWhile \label{code:our_contribution_2}
	\State $\pi = \text{REMOVE}(x_r, x_s)$\label{code:initial_remove_comp}
	\While {$\pi = \emptyset$}\label{code:our_contribution_3}
		\State Update $\rho = \min_{\{x_{r'}, x_{s'}\}\not\in E}d(x_{r'}, x_{s'})$ and $E\gets E\cup \{(x_{r'},x_{s'})\}$\label{code:add_edge_for_algorithm_remove}
		\State $\pi = \text{REMOVE}(x_r, x_s)$\label{code:run_algorithm_remove}
	\EndWhile \label{code:end_while_loop_remove}
	\State $L\gets L\setminus \pi$ \label{code:our_contribution_4}
	\State \Return $L, \rho$
	\end{algorithmic} 
\end{algorithm}

\begin{algorithm}
\caption{REMOVE}
\label{alg:landmark_removal_condition}
	\hspace*{\algorithmicindent} \textbf{Input} $x_r, x_s$
	\begin{algorithmic}[1]
	\If {$x_r \in L$}
		\If {$OnlyBy(x_r) = \emptyset$} 
			\State\Return $\{x_r\}$
		\ElsIf {$x_s \in L \text{ and } OnlyBy(x_s) = \emptyset$}
			\State\Return $\{x_{s}\}$
		\Else
			\State\Return $\emptyset$
		\EndIf
	\Else
		\If {$OnlyBy(x_s) = \emptyset$}
			\State\Return $\{x_s\}$
		\Else
			\State\Return $\emptyset$
		\EndIf
	\EndIf
	\end{algorithmic} 
\end{algorithm}

\begin{algorithm}
\caption{Two-way merge sorting algorithm~\cite{knuth1998art}}
\label{alg:linear_sort}
	\hspace*{\algorithmicindent} \textbf{Input} $X = \left(X[1],\ldots, X[p]\right), Y = \left(Y[1],\ldots, Y[q]\right)$, where $X[r']\le X[r'+1]$ $\forall r'\in\{1,\ldots, p-1\}$, $Y[s']\le Y[s'+1]$ $\forall s' \in \{1,\ldots, q-1\}$
	\begin{algorithmic}[1]
	\State Initialize an empty array $Z = ()$.
	\State $r = 1, s = 1$\label{code:initial_sort_index}
	\While {$r < p+1 \text{ and } s < q+1$}\label{code:start_while_sort}
	\If {$X[r] \le Y[s]$}\label{code:if_condition_while_sort}
		\State $Z.\text{append}\left(X[r]\right)$\label{code:start_if_condition_while_sort}
		\State $r = r + 1$\label{code:end_if_condition_while_sort}
	\Else\label{code:else_condition_while_sort}
		\State $Z.\text{append}\left(Y[s]\right)$\label{code:start_else_condition_while_sort}
		\State $s = s + 1$\label{code:end_else_condition_while_sort}
	\EndIf
	\EndWhile\label{code:end_while_sort}
	\If {$r = p+1$}
		\State\Return $\text{Concatenate}\left(Z, \left(Y[s],\ldots, Y[q]\right)\right)$\label{code:concat_first_case}
	\Else
		\State\Return $\text{Concatenate}\left(Z,\left(X[r],\ldots,X[p]\right)\right)$\label{code:concat_second_case}
	\EndIf
	\end{algorithmic} 
\end{algorithm}

\section{Mathematical properties\label{sec:math}}

In this section, we show some mathematical properties related to the proposed landmark replacement algorithm, in particular its time complexity.

\begin{theorem}{\cite{parekh1991analysis}}\label{thm:upperbound_greedy}
For an undirected graph $(V,E)$,
\begin{equation}
\left| L\right| \le \left| V\right| + 1 - \sqrt{2\left| E\right| + 1},
\end{equation}
where $L$ is the dominating set of $(V,E)$ obtained from Algorithm~\ref{alg:greedy_ds}.
\end{theorem}

\noindent The following Corollary of Theorem~\ref{thm:upperbound_greedy} tells us how we should set the value of $\rho$ to guarantee that the number of landmarks obtained by Algorithm \ref{alg:greedy_ds} is at most $m$.
\begin{corollary}\label{corr:upperbound_greedy}
Consider an undirected graph $(V,E)$, where $V=\{x_1,\ldots,x_n\}$. If we set the value of parameter $\rho$ such that
\begin{equation}
\left| E\right| = \left\lceil \frac{(n - m + 1)^{2} - 1}{2}\right\rceil,
\label{eq:E-for-corollary}
\end{equation}
where $E = \{(x_r, x_s): d(x_r, x_s) < \rho\}$, then $L$ obtained by Algorithm \ref{alg:greedy_ds} satisfies $\left| L\right| \le m$.
\end{corollary}
\begin{proof}
Using Eq.~\eqref{eq:E-for-corollary}, we obtain
\begin{align}
\sqrt{2\left| E\right| + 1} &=  \sqrt{2\left(\left\lceil \frac{(n - m + 1)^{2} - 1}{2}\right\rceil\right) + 1}\notag \\
&\ge \sqrt{2\left(\frac{(n - m + 1)^{2}- 1}{2}\right) + 1}\notag \\
&= n - m + 1.\label{eq:upper_bound_corr_1}
\end{align}
Equation~\eqref{eq:upper_bound_corr_1} implies that
\begin{equation}
n + 1 - \sqrt{2\left| E\right| + 1}\le m.
\label{eq:upper_bound_corr_2}
\end{equation}
By combining Theorem~\ref{thm:upperbound_greedy} and Eq.~\eqref{eq:upper_bound_corr_2}, we obtain
\begin{equation}
\left| L\right| \le n + 1 - \sqrt{2\left| E\right| + 1}\le m.
\end{equation}
\end{proof}

\begin{theorem}\label{thm:neighbor_degree}
Let $x$ be a node in an undirected graph $(V, E)$. If $\deg(x) > 2\sqrt{\left| E\right|}$, then there is a neighbor $x'$ of $x$ such that $\deg(x')\le\sqrt{\left|E\right|}$.
\end{theorem}
\begin{proof}
Suppose to the contrary that $\deg(x')>\sqrt{\left|E\right|}$ for all $x'\in\mathcal{N}(x)$. Therefore, we obtain
\begin{align}
2\left| E\right| &\ge \sum_{x'\in\mathcal{N}(x)} \deg(x')\notag\\
&>  \sum_{x'\in\mathcal{N}(x)} \sqrt{\left| E\right|}\notag\\
&= \sqrt{\left| E\right|} \cdot \left| \mathcal{N}(x)\right| \notag\\
&> 2\left|E\right|,
\end{align}
which is a contradiction.
\end{proof}
\begin{remark}
This theorem extends a theorem given in \cite{hjuler2019dominating}, which showed that $\deg(x) > 2 \sqrt{\left|E\right|} + 1$ leads to $\deg(x') > \sqrt{\left|E\right|}$.
\end{remark}
\begin{remark}
A proof of the theorem in \cite{hjuler2019dominating}, which is omitted in \cite{hjuler2019dominating}, is similar to this proof.
\end{remark}

\begin{lemma}
Assume that we have constructed the adjacency list of graph $(V,E)$, where $V = \{x_1,\ldots, x_n\}$, and $E = \{(x_r, x_s) : d(x_r, x_s) < \rho \}$, where $r, s\in\{1,\ldots, n\}$; the adjacency list is the list of neighbors of each node. Then, the time complexity for computing $OnlyBy(z)$ for a given $z\in L$ is $O(mn)$.
\label{lemma:onlyby_time_complexity}
\end{lemma}
\begin{proof}
First, if $z$ is an isolated node, then $\mathcal{N}(z)\cup\{z\} = \{z\}$ and $\left| \left(\mathcal{N}(z)\cup\{z\}\right)\cap L\right| = \left| \{z\}\right| = 1$. Therefore, we obtain $OnlyBy(z) = \{z\}$, and the time complexity for computing $OnlyBy(z)$ is $O(1)$. Now, we assume that $z$ is not isolated. Then, we check if $\mathcal{N}(z) \subseteq L$. To do so, we examine the list of neighbors of $z$. If $\mathcal{N}(z) \subseteq L$, then every neighbor of $z$ is a landmark, and we obtain $OnlyBy(z) = \emptyset$. The time complexity for checking whether or not $\mathcal{N}(z) \subseteq L$ is
\begin{equation}
O\left(\left|L\right|\deg(z)\right)\le O\left(mn\right).\label{eq:time_complexity_neighbor_subset_landmark}
\end{equation}

If $\mathcal{N}(z)\nsubseteq L$, then we create two arrays of length $n$, denoted by $A = (A[1], \ldots, A[n])$ and $B = (B[1], \ldots, B[n])$ and initialize each entry of $A$ and $B$ by zero. Then, we set $A[r]= 1$ for each $x_r \in\mathcal{N}(z)\cup\{z\}$. Similarly, we set $B[r] = 1$ for each $x_r$ satisfying $x_r \in \mathcal{N}(z')\cup\{z'\}$ for any $z' \in L\setminus \{z\}$.
The time complexity of setting $A$ and $B$ is $O(n)$ and $O(mn)$, respectively.
%
%
Because every node is covered by a landmark, either $A\left[r \right]=1$ or $B\left[r \right]=1$ holds true for all $r \in\{1,\ldots, n\}$. Therefore, there are three cases to consider: $(A\left[r \right], B\left[r \right]) = (1, 1)$, $(1, 0)$, and $(0, 1)$.

If $(A\left[r \right], B\left[r \right]) = (1, 1)$, then node $x_r$ is covered by both $z$ and at least one node in $L\setminus\{z\}$. In this case, we do not add $x_r$ to $OnlyBy(z)$. If $(A\left[r \right], B\left[r \right]) = (1, 0)$, then we add $x_r$ to $OnlyBy(z)$ because $z$ is the only landmark that covers $x_r$. If $(A\left[r \right], B\left[r\right]) = (0, 1)$, then we do not add $x_r$ to $OnlyBy(z)$ because $x_r$ is not a neighbor of $z$ while $x_r$ is covered by at least one node in $L\setminus\{z\}$. Therefore, the total time complexity for constructing $OnlyBy(z)$ from the arrays $A$ and $B$ is $O(n)$. 

By combining these time complexities, we find that the computation of $OnlyBy\left(z\right)$ requires
$O(1) + O\left(mn\right) + O\left(n\right) + O\left(mn\right) + O\left(n\right) = O\left(mn\right)$
time.
\end{proof}

We now analyze the time complexity for executing Algorithm~\ref{alg:landmark_replacement} upon the addition of a single new data point, $x_{n+1}$. We assume that the pairwise distance between the existing points has been calculated and sorted and that the pairwise distance between each existing point and $x_{n+1}$ has been calculated but not sorted yet. We proceed by sorting the array of pairwise distance between each existing point and $x_{n+1}$. Then, we apply Algorithm~\ref{alg:linear_sort} to the two sorted arrays to obtain the sorted array of the pairwise distances of all points including $x_{n+1}$. Algorithm~\ref{alg:linear_sort}, known as the two-way merge sorting algorithm, runs in $O(p+q)$ time, where $p$ and $q$ are the lengths of each input array~\cite{knuth1998art}.
\begin{theorem}
\label{thm:time_complexity_single_arrival}
Upon the arrival of $x_{n+1}$, we assume that $d(x_1, x_{n+1})$, $\ldots$, $d(x_n, x_{n+1})$ have been computed, that $\{ d(x_1, x_2), d(x_1, x_3), \ldots, d(x_{n-1}, x_n) \}$ has been sorted, and that we have updated the adjacency list. Then, the time complexity of executing Algorithm~\ref{alg:landmark_replacement}
with $V=\{x_1,\ldots, x_{n+1}\}$ and $E = \{(x_r, x_s) : r, s \in \{1, \ldots, n+1\}, d(x_r, x_s) < \rho \}$ as input is $O\left(n^2 + e_{n+1}mn\right)$, where $e_{n+1}$ is the number of edges added upon executing Algorithm~\ref{alg:landmark_replacement}, and $m = \left|L\right|$, i.e., the number of landmarks.
\end{theorem}

\begin{proof}
First, we sort the array $\{ d(x_1, x_{n+1}), \ldots, d(x_n, x_{n+1}) \}$, which requires $O(n\log(n))$ time. Then, we use Algorithm~\ref{alg:linear_sort} to merge sort the sorted array of $\{ d(x_1, x_{n+1}), \ldots, d(x_n, x_{n+1}) \}$ and the sorted array of $\{ d(x_r, x_s) : r, s \in \{1, \ldots, n \}, r \neq s \}$. This merge sort requires $O(n^2)$ time because the length of the second array is $O(n^2)$. To summarize these two steps, we have sorted $d(x_r, x_s)$ with $\forall r, \forall s \in \{1, \ldots, n+1 \}$, $r \neq s$ in
$C_1 = 
O(n\log n) + O(n^2) = O(n^2)$
time. 

Theorem~\ref{thm:neighbor_degree} implies that one only needs to inspect at most $2\left\lceil\sqrt{\left|E\right|}\right\rceil +1$ neighbors of $x_{n+1}$ to be able to obtain a neighbor whose degree is less than or equal to $\sqrt{\left| E\right|}$. Therefore, we can carry out line~\ref{code:search_neighbor_x_n+1} of Algorithm~\ref{alg:landmark_replacement} in $O\left(\sqrt{\left| E\right|}\right)$ time. Because we have computed and sorted the pairwise distances between the nodes, lines~\ref{code:our_contribution_1}--\ref{code:first_epsilon_update} run in $O\left(1\right)$ time. Therefore, the time complexity of carrying out lines~\ref{code:search_neighbor_x_n+1}--\ref{code:first_epsilon_update} is 
$C_2 =
O\left(\sqrt{\left| E\right|}\right)$.
We denote by $e^{(1)}_{n+1}$ the number of edges that one adds to the network in lines~\ref{code:start_first_while_loop}--\ref{code:our_contribution_2}. Moreover, we update the adjacency list upon adding $e^{(1)}_{n+1}$ edges to the network. Therefore, the total cost of carrying out lines~\ref{code:start_first_while_loop}--\ref{code:our_contribution_2} is
$C_3 =
O\left(e^{(1)}_{n+1}\right) + O\left(2e^{(1)}_{n+1}\right) = O\left(e^{(1)}_{n+1}\right)$,
where $O\left(2e^{(1)}_{n+1}\right)$ comes from the time complexity for updating the adjacency list. Lemma~\ref{lemma:onlyby_time_complexity} guarantees that line~\ref{code:initial_remove_comp} runs in 
$C_4 =
O\left(m(n+1)\right)$
time. 

In lines~\ref{code:our_contribution_3}--\ref{code:end_while_loop_remove}, for each added edge, we update the adjacency list and then call Algorithm~\ref{alg:landmark_removal_condition}. Algorithm~\ref{alg:landmark_removal_condition} requires that we (i) check whether or not one or both input nodes belong to $L$ and then (ii) compute an $OnlyBy$. The cost of (i) is $O\left(m\right)$. Lemma~\ref{lemma:onlyby_time_complexity} indicates that the cost of (ii) is $O(m(n+1))$. The total time complexity for updating the adjacency list and carrying out lines~\ref{code:our_contribution_3}--\ref{code:end_while_loop_remove} is
$C_5 =
O\left(2e^{(2)}_{n+1}\right) + O\left(e^{(2)}_{n+1}m\right) + O\left(e^{(2)}_{n+1}m(n+1)\right)  = O\left(e^{(2)}_{n+1}m(n+1)\right)$,
where $e^{(2)}_{n+1}$ is the number of edges added upon the execution of lines~\ref{code:our_contribution_3}--\ref{code:end_while_loop_remove}. 

Finally, we
add $C_1$, $C_2$, $C_3$, $C_4$, and $C_5$,
and use $\left| E\right| \le n(n-1)/2$ to find the total time complexity of executing Algorithm~\ref{alg:landmark_replacement} as follows:
\begin{align}
O\left(n^2 + \sqrt{\left| E\right|} + e^{(1)}_{n+1} + m(n+1) + e^{(2)}_{n+1}m(n+1)\right) &= O\left(n^2 + e^{(1)}_{n+1} + e^{(2)}_{n+1}m(n+1)\right)\notag\\
&\le O\left(n^2 + \left[e^{(1)}_{n+1} + e^{(2)}_{n+1}\right]m(n+1)\right)\notag\\
&= O\left(n^2 + e_{n+1}mn\right).
\label{eq:time_complexity_single_arrival}
\end{align}
\end{proof}

\begin{corollary}
Suppose that we sequentially add $x_{n+1},x_{n+2},\ldots, x_{n'}$ with possible landmark replacements
by iterating Algorithm~\ref{alg:landmark_replacement} with $\left| L\right| = m$. We further assume that the distance between the new data point, $x_{\ell}$, and each existing point, i.e., $d(x_r, x_\ell)$, $r \in \{1, \ldots, \ell-1\}$, has been computed upon the addition of $x_\ell$ and that $\{ d(x_1, x_2), d(x_1, x_3), \ldots, d(x_{\ell-2}, d_{\ell-1}) \}$ has been sorted. Then, the time complexity for iterating Algorithm~\ref{alg:landmark_replacement} over $x_{n+1},x_{n+2},\ldots, x_{n'}$ is $O\left(m{n'}^3\right)$.
\label{corr:time_complexity_sequential_arrival}
\end{corollary}
\begin{proof}
For each new $x_\ell$ with $\ell\in\{n+1,\ldots, n'\}$, Theorem~\ref{thm:time_complexity_single_arrival} implies that running Algorithm~\ref{alg:landmark_replacement} requires $O\left((\ell-1)^2 + e_\ell m(\ell-1)\right)=O\left(\ell^2 + e_\ell m\ell\right)$ time. Therefore, the time complexity of iterating Algorithm~\ref{alg:landmark_replacement} from $x_{n+1}$ to $x_{n'}$ is
\begin{align}
\sum_{\ell =n+1}^{n'} O\left(\ell^2 + e_\ell m\ell \right) 
%
%
&\le \sum_{\ell=n+1}^{n'} O\left(\ell^2\right) + O\left(mn'\right)\sum_{\ell=n+1}^{n'} O\left(e_\ell\right)\notag\\
&\le O\left({n'}^3\right) + O\left(mn'\right)O\left(\binom{n'}{2}\right)\notag\\
&= O\left(m{n'}^3\right).
\end{align}
\end{proof}

To summarize this section, we first related the distance threshold for the geometric graph, $\rho$, and the number of landmarks allowed, $m$, for 
Algorithm~\ref{alg:greedy_ds}, which computes the dominating set in an offline manner
(Theorem~\ref{thm:upperbound_greedy} and Corollary~\ref{corr:upperbound_greedy}). In the remainder of this section, we mathematically derived the time complexity of our online landmark replacement algorithm
(i.e., Algorithm~\ref{alg:landmark_replacement}),
in Theorem~\ref{thm:time_complexity_single_arrival} for a single landmark replacement and Corollary~\ref{corr:time_complexity_sequential_arrival} for sequential replacements.

\section{Landmark multidimensional scaling}

The landmark replacement method proposed in section~\ref{sub:proposed_online_landmark} works for any out-of-sample dimensionality reduction method as long as the distance between each pair of data points in the original space is defined. In the remainder of this article, we use the LMDS to evaluate the proposed landmark replacement method mathematically and numerically.

\subsection{Algorithms}

In this section, we briefly review the LMDS \cite{de2004sparse}, which is an out-of-sample extension of the classical MDS.

The classical MDS works as follows. We consider the set of $n$ data points to be embedded, $S= \{ x_1, \ldots, x_n \}$, with the distance between two data points being given by $d(x_r, x_s)$. The MDS is a map from $S$ to $\mathbb{R}^k$. To construct the MDS, we denote the $n \times n$ squared distance matrix by $\Delta = [\Delta_{rs}]$, where $\Delta_{rs} = d^2(x_r,x_s)$. We then define the double mean-centered dot product matrix by $D = -\frac{1}{2}H_n\Delta H_n$, where $H_n = I_n-\frac{1}{n}J_n$, matrix $I_n$ is the $n\times n$ identity matrix, and $J_n$ is the $n\times n$ matrix in which all the entries are equal to 1. The coordinates of $x_r$ in $\mathbb{R}^k$ after the mapping is the $r$th column of
\begin{equation}
    L_k = \begin{bmatrix}
    \sqrt{\lambda_1}v_1^{\top}\\
    \vdots \\
    \sqrt{\lambda_k}v_k^{\top}
    \end{bmatrix},
    \label{cMDS_coordinates}
\end{equation}
which is a $k \times n$ matrix. In Eq.~\eqref{cMDS_coordinates},
$\lambda_r$ is the $r$th largest eigenvalue of $D$, which we assume to be positive, $v_r$ is the right eigenvector associated with eigenvalue $\lambda_r$, normalized in terms of the $\ell^2$-norm, and ${}^\top$ represents the transposition.
The time complexity of the MDS is $O(n^3)$ owing to the eigenvalue and eigenvector computation~\cite{gisbrecht2015metric}. Note that the MDS preserves the pairwise distance between all pairs of data points in the Euclidean embedding space if and only if matrix $D = -\frac{1}{2}H_n\Delta H_n$ is positive semidefinite \cite{graepel1999classification, pekalska2001generalized}.

The LMDS only applies the MDS to a subset of the entire data set called the landmarks \cite{de2004sparse, Desilva2003Nips, Platt2005Aistats}. The LMDS then determines the mapping of other arbitrary data points with the help of the landmarks. We denote by $L =\{x_{i_1},\ldots, x_{i_m}\}\subseteq S$ the set of landmarks and apply the MDS to any $x \in L$ using Eq.~\eqref{cMDS_coordinates}. The coordinate in $\mathbb{R}^k$ of an arbitrary out-of-sample data point $x\notin L$ is given by
\begin{equation}
    \psi(x) = - \frac{1}{2}L'_k(\delta_x - \delta_\mu),
    \label{LMDS_arbitrary_coordinates}
\end{equation}
where
\begin{equation}
    L'_k = \begin{bmatrix}
    \frac{1}{\sqrt{\lambda_1}}v_1^{\top}\\
    \vdots \\ 
    \frac{1}{\sqrt{\lambda_k}}v_k^{\top}
    \end{bmatrix} \in \mathbb{R}^{k\times m},
    \label{eq:LMDS_eigen_projection}
\end{equation}
\begin{equation}
    \delta_{x} = \begin{bmatrix}
    d^2(x,x_{i_1})\\ \
    \vdots \\ 
    d^2(x,x_{i_m})
    \end{bmatrix},
\end{equation}
\begin{equation}
        \delta_\mu = \frac{1}{m}\sum_{r=1}^{m} \begin{bmatrix}
    d^2(x_{i_1},x_{i_r})\\ \
    \vdots \\ 
    d^2(x_{i_m},x_{i_r})
    \end{bmatrix}.
        \label{LMDS_avg_coordinates}
\end{equation}

\subsection{Goodness of embedding}

To assess the quality of the LMDS, we use the following two goodness of fit criteria.
First, the normalized stress function~\cite{borg2005modern, duin2005dissimilarity, kruskal1964multidimensional} is defined by
\begin{equation}
\sigma = \sqrt{\frac{\sum_{r=1}^n \sum_{s = 1}^{r - 1} \left[ d(x_{s},x_{r}) - \left|\psi(x_{s}) - \psi(x_{r}) \right|_2 \right]^2}
{\sum_{r=1}^n \sum_{s = 1}^{r-1} d^2(x_{s},x_r)}}\,.
\label{eq:stress_function}
\end{equation}
Note that $\left|\psi(x_{s}) - \psi(x_{r}) \right|_2$ is the Euclidean distance between the two data points in the embedding space. One obtains a smaller $\sigma$ if the MDS better preserves the distance between pairs of data points in the original space. A guideline of acceptable $\sigma$ values is $\sigma < 0.15$~\cite{borg2005modern}.

Second, we denote by $\sigma_L$ the normalized stress computed only on the basis of the set of landmarks. A small $\sigma_L$ value indicates that the MDS well preserves the distance structure for the set of $m$ landmarks.
%

\section{Dislocation of the embedding coordinates upon landmark replacement\label{sec:displacement}}

In this section, we mathematically evaluate how much a data point moves in the embedding space when one replaces a landmark with another data point. We provide relevant basic definitions and theorems in electronic supplementary material section~\ref{sec:linear_operators}.

\begin{definition}
Given a pair of $a\times b$ real-valued matrices $P=[p_{rs}]$ and $Q=[q_{rs}]$, the Frobenius inner product is defined by
\begin{equation}
\left\langle P,Q\right\rangle_F = \sum_{r=1}^{a}\sum_{s=1}^{b}p_{rs}q_{rs}.
\end{equation}
Additionally, we define the induced Frobenius norm by
\begin{equation}
\left\| P\right\|_F = \sqrt{\left\langle P,P\right\rangle_F}
\end{equation}
and the Frobenius distance by
\begin{equation}
d_F(P,Q) = \left\|P-Q\right\|_F.
\end{equation}
\end{definition}

\begin{theorem} (De Silva and Tenenbaum \cite{de2004sparse})
Let $\Delta$ be the squared-distance matrix of $m$ landmarks, and let $\lambda_1\ge \lambda_2\ge\cdots\ge \lambda_m$ be the eigenvalues of $D = -\frac{1}{2}H_m\Delta H_m$ with $\lambda_k > 0$ and $\lambda_k > \lambda_{k+1}$, where $k$ is the dimension of the embedding space, $\mathbb{R}^k$. Consider a perturbation of $\Delta$ given by $\tilde{\Delta} = \Delta + \varepsilon\Phi + O(\varepsilon^2)$, where $| \varepsilon | \ll 1$. Then, the $L'_k$ matrix after the perturbation, denoted by $\tilde{L}'_k$, is written as
\begin{equation}
\tilde{L}'_k = L'_k + \varepsilon\varphi + O(\varepsilon^2),
\end{equation}
where $\varphi \in \mathbb{R}^{k \times m}$ satisfies
\begin{equation}
\left\|\varphi\right\|_F \le \left[\frac{1}{4\lambda_k^{3/2}} + \frac{1}{2\lambda_k^{1/2}(\lambda_k - \lambda_{k+1})}\right]\left\|\Phi\right\|_F.
\end{equation}
%
\label{thm:eigen_projection_bound}
\end{theorem}

Replacing an existing landmark by a new one results in a change in a column in $\Delta$. 
Theorem~\ref{thm:eigen_projection_bound} provides an upper bound of the size of perturbation in $L'_k$ in response to such a change in $\Delta$.
In the following Theorem~\ref{thm:embedding_dislocation}, we exploit
Theorem~\ref{thm:eigen_projection_bound} to bound the size of dislocation of the embedding coordinate of an arbitrary data point upon a landmark replacement. The following Lemmas~\ref{lemma:bound_diff_sq_distance}, \ref{lemma:bound_first_part}, and \ref{lemma:diff_avg_column_bound} aid the proof of Theorem~\ref{thm:embedding_dislocation}. Note that Theorems~\ref{thm:bound_range_operator} and \ref{thm:bounded_operator} also assist minor details in the proof of Theorem~\ref{thm:embedding_dislocation}.

\begin{lemma}
Let $(X,d)$ be a metric space and $x,y,z\in X$. Assume that $d(y,z) < \max\{d(x,y),d(x,z)\}$. Then, we obtain
\begin{equation}
\left|d^2(x,y) - d^2(x,z)\right| \le 2d(x, y)d(y, z) + d^2(y,z).
\label{eq:bound_abs}
\end{equation}
\label{lemma:bound_diff_sq_distance}
\end{lemma}
\begin{proof}
The triangle inequality gives $d(x,z)\le d(x,y) + d(y,z)$. Therefore,
\begin{equation}
d^2(x,y) -  d^2(x,z) \ge - 2d(x,y)d(y,z) - d^2(y,z).
\label{eq:triangle_upper_bd}
\end{equation}
Next, we consider another triangle inequality $d(x,y)\le d(x,z) + d(y,z)$, which leads to $d(x,y) - d(y,z)\le d(x,z)$. Because we assumed that $d(y, z) < \max\{d(x,y),d(x,z)\}$, we obtain $d(y, z) \le d(x, y)$, without the loss of generality. Therefore, we obtain
$d^2(x,y) - 2d(x,y)d(y,z) + d^2(y,z) \le d^2(x, z)$, i.e.,
\begin{equation}
d^2(x,y) - d^2(x, z) \le 2d(x,y)d(y,z) - d^2(y, z) \le 2d(x,y)d(y,z) + d^2(y, z).
\label{eq:triangle_upper_bd-2}
\end{equation}
Combination of Eqs.~\eqref{eq:triangle_upper_bd} and \eqref{eq:triangle_upper_bd-2} yields Eq.~\eqref{eq:bound_abs}.
\end{proof}

\begin{lemma}
Let $(X,d)$ be a metric space. Let $\Delta_1$ and $\Delta_2$ be the squared distance matrices associated with the sets of landmarks $\{x_1,\ldots, x_{m-1}, x_m\}$ and $\{x_1,\ldots, x_{m-1},y\}$, respectively, where $x_1,\ldots, x_m, y\in X$, and $\varepsilon \equiv d(x_m, y)$. We further let $L'^{(1)}_k$ and $L'^{(2)}_k$ be $k\times m$ matrices defined by Eq.~\eqref{eq:LMDS_eigen_projection} and associated with $\{x_1,\ldots, x_{m-1}, x_m\}$ and $\{x_1,\ldots, x_{m-1},y\}$, respectively. Furthermore, let $\lambda_1\ge\cdots\ge\lambda_k$ be the eigenvalues associated with $\Delta_1$ such that $\lambda_k >0$ and $\lambda_k > \lambda_{k+1}$. Then, there is a perturbation $\Delta_2 = \Delta_1 + \varepsilon\Phi + O(\varepsilon^2)$ for small $\varepsilon$ such that $\left\|\Phi\right\|_F \le 2\sqrt{2K(m-1)}$, where $K=\max\{d^2(x_1,x_m),\ldots, d^2(x_{m-1},x_m)\}$. Moreover, if $z\in X\setminus\{x_1,\ldots, x_m,y\}$ and $\varepsilon \le d(z,y)$, then we obtain
\begin{equation}
 \left| L'^{(2)}_k\delta^{(2)}_z - L'^{(1)}_k\delta^{(1)}_z\right|_2 \le 2\varepsilon \left[d(z,y)\left\|L'^{(2)}_k\right\|_2 + \widetilde{K}\sqrt{2K(m-1)} \left|\delta^{(1)}_z\right|_2\right] + O(\varepsilon^2),
 \label{eq:embed_bd_first_quant}
 \end{equation}
where $\delta^{(1)}_z = \begin{bmatrix} d^2(z,x_1)\\ \vdots\\ d^2(z,x_{m-1})\\ d^2(z,x_m)] \end{bmatrix}$, $\delta^{(2)}_z = \begin{bmatrix} d^2(z,x_1)\\ \vdots\\ d^2(z,x_{m-1})\\ d^2(z,y)\end{bmatrix}$, and $\widetilde{K} = \frac{1}{4\lambda_k^{3/2}} + \frac{1}{2\lambda_k^{1/2}(\lambda_k - \lambda_{k+1})}$. Here, $\left\|\cdot\right\|_2$ denotes the matrix 2-norm, which we obtain by substituting $\left(X_1,\left\|\cdot\right\|_{X_1}\right) = \left(\mathbb{R}^m, \left|\cdot\right|_2\right)$ and $\left(X_2,\left\|\cdot\right\|_{X_2}\right) = \left(\mathbb{R}^k, \left|\cdot\right|_2\right)$ in Definition~\ref{def:operator_norm} in electronic supplementary material section~\ref{sec:linear_operators}.
\label{lemma:bound_first_part}
\end{lemma}
\begin{proof}
First, we obtain
\begin{equation}
\Delta_2 = \Delta_1 + \left[\renewcommand{\arraystretch}{1.2}\begin{array}{c|c}
\bm{0}_{(m-1)\times (m-1)} & u\\
\hline
u^{\top} & 0
\end{array}\right],
\label{eq:perturbed_sq_dist_matrix}
\end{equation}
where $\bm{0}_{(m-1)\times (m-1)}$ is the zero matrix of size $(m-1)\times (m-1)$, and
\begin{equation}
u = \begin{bmatrix} d^2(x_1,y) - d^2(x_1,x_m) \\ d^2(x_2,y) - d^2(x_2,x_m) \\ \vdots\\ d^2(x_{m-1},y) - d^2(x_{m-1},x_m)\end{bmatrix}.
\label{eq:def-u}
\end{equation}
For each $r \in\{1,\ldots,m-1\}$, we map $x_r\mapsto \tilde{x}_r,x_m\mapsto \tilde{x}_m$ and $y\mapsto \tilde{y}$, where $\tilde{x}_r, \tilde{x}_m, \tilde{y}$ are vertices of a triangle, whose sides are of lengths $d(x_r,x_m), d(x_r,y)$, and $\varepsilon$, in $\mathbb{R}^2$. This is doable because the metric $d$ satisfies the triangle inequality. By the law of cosine, we obtain
\begin{equation}
d^2(x_r, y) - d^2(x_r, x_m) = -2d(x_r,x_m)\cos(\theta_r) \varepsilon + \varepsilon^2,
\label{eq:sq_dist_perturbation_bigO}
\end{equation}
where $\theta_r$ is the angle between $\overrightarrow{\tilde{x}_m\tilde{x}_r}$ and $\overrightarrow{\tilde{x}_m\tilde{y}}$.
By substituting Eq.~\eqref{eq:sq_dist_perturbation_bigO} in Eq.~\eqref{eq:def-u}, we obtain
\begin{equation}
u = \begin{bmatrix}
\phi_1\\ \vdots \\ \phi_{m-1}
\end{bmatrix}\varepsilon + O(\varepsilon^2),
\label{eq:perturbed_u}
\end{equation}
where
\begin{equation}
\phi_r \equiv -2d(x_r,x_m)\cos(\theta_r),\quad r\in\{1,\ldots, m-1\}.
\label{eq:def-phi_r}
\end{equation}
Then, we define
\begin{equation}
\Phi = \begin{bmatrix}
0 & 0 & \ldots & 0 & \phi_1\\
0 & 0 & \ldots & 0 & \phi_2\\
\vdots & \vdots & \ddots & \vdots & \vdots\\
0 & 0 & \dots & 0 & \phi_{m-1}\\
\phi_1 & \phi_2 & \ldots & \phi_{m-1} & 0
\end{bmatrix}.
\label{eq:define_big_phi}
\end{equation}
Substituting Eqs.~\eqref{eq:perturbed_u}, \eqref{eq:define_big_phi} into Eq.~\eqref{eq:perturbed_sq_dist_matrix} yields
\begin{equation}
\Delta_2 = \Delta_1 + \varepsilon\Phi + O(\varepsilon^2).
\label{eq:perturb_dist_matrix_wtih_epsilon}
\end{equation}

By substituting $\left(X_1,\left\langle \cdot,\cdot \right\rangle_{X_1}\right) = \left(\mathbb{R}^m, \left\langle \cdot,\cdot \right\rangle_{E}\right)$ and $\left(X_2,\left\langle \cdot,\cdot \right\rangle_{X_2}\right) = \left(\mathbb{R}^k, \left\langle \cdot,\cdot \right\rangle_{E}\right)$ in Theorem~\ref{thm:bounded_operator}, where $ \left\langle \cdot,\cdot \right\rangle_{E}$ denotes the Euclidean inner product, it follows that $L'^{(2)}_k - L'^{(1)}_k$ and $L'^{(2)}_k$ are bounded. Using Theorem~\ref{thm:bound_range_operator}, we obtain
\begin{align}
\left| L'^{(2)}_k\delta^{(2)}_z - L'^{(1)}_k\delta^{(1)}_z\right|_2 &= \left| L'^{(2)}_k\delta^{(2)}_z - L'^{(2)}_k\delta^{(1)}_z + L'^{(2)}_k\delta^{(1)}_z - L'^{(1)}_k\delta^{(1)}_z\right|_2\notag\\
&\le \left| L'^{(2)}_k\delta^{(2)}_z - L'^{(2)}_k\delta^{(1)}_z\right|_2 + \left| L'^{(2)}_k\delta^{(1)}_z - L'^{(1)}_k\delta^{(1)}_z\right|_2\notag\\
&\le \left\|L'^{(2)}_k\right\|_2 \left|\delta^{(2)}_z - \delta^{(1)}_z\right|_2 + \left\|L'^{(2)}_k - L'^{(1)}_k\right\|_2 \left|\delta^{(1)}_z\right|_2\notag\\
&\le \left\|L'^{(2)}_k\right\|_2 \left|\delta^{(2)}_z - \delta^{(1)}_z\right|_2 + \left\|L'^{(2)}_k - L'^{(1)}_k\right\|_F \left|\delta^{(1)}_z\right|_2\notag\\
&=  \left\|L'^{(2)}_k\right\|_2 \left|\begin{bmatrix}0 \\ \vdots \\ 0 \\ d^2(z,y) - d^2(z,x_m)\end{bmatrix}\right|_2 + \left\|L'^{(2)}_k - L'^{(1)}_k\right\|_F \left|\delta^{(1)}_z\right|_2.
\label{eq:embedding_main_bound_first_part_subs}
\end{align}
Lemma~\ref{lemma:bound_diff_sq_distance} implies
\begin{equation}
\left| d^2(z,y) - d^2(z,x_m)\right| \le 2d(z,y)\varepsilon + O(\varepsilon^2).
\label{eq:proof-bound-derivation-a}
\end{equation}
Because $\Delta_2 = \Delta_1 + \varepsilon\Phi + O(\varepsilon^2)$, by Theorem~\ref{thm:eigen_projection_bound}, we can write
\begin{equation}
L'^{(2)}_k = L'^{(1)}_k + \varepsilon\varphi + O(\varepsilon^2),
\label{eq:proof-bound-derivation-b}
\end{equation}
where
\begin{equation}
\left\|\varphi\right\|_F \le \widetilde{K}\left\|\Phi\right\|_F.
\label{eq:proof-bound-derivation-c}
\end{equation}
By substituting Eqs.~\eqref{eq:proof-bound-derivation-a}, \eqref{eq:proof-bound-derivation-b}, and \eqref{eq:proof-bound-derivation-c} in Eq.~\eqref{eq:embedding_main_bound_first_part_subs}, we obtain
\begin{align}
(\text{RHS of Eq.~}\eqref{eq:embedding_main_bound_first_part_subs})
\le& \left\|L'^{(2)}_k\right\|_2\left[ 2d(z,y)\varepsilon + O(\varepsilon^2)\right] + \varepsilon\left\|\varphi\right\|_F \left|\delta^{(1)}_z\right|_2 + O(\varepsilon^2)\notag\\
%
%
\le& \varepsilon\left[2d(z,y)\left\|L'^{(2)}_k\right\|_2 + \widetilde{K}\left\|\Phi\right\|_F \left|\delta^{(1)}_z\right|_2\right] + O(\varepsilon^2).
\label{eq:before_final_bound_first_part}
\end{align}
To bound $\left\|\Phi\right\|_F$, we use the fact that $\phi^2_r \le 4d^2(x_r,x_m)$, which obeys from Eq.~\eqref{eq:def-phi_r}, to proceed as follows:
\begin{equation}
\left\|\Phi\right\|_F = \sqrt{\sum_{r=1}^{m-1} 2\phi^2_r} \le \sqrt{\sum_{r=1}^{m-1}8d^2(x_r,x_m)} \le \sqrt{\sum_{r=1}^{m-1}8K} = 2\sqrt{2K(m-1)},
\label{eq:bound_big_phi}
\end{equation}
where $K$ is defined in the statement of the present Lemma. Using Eq.~\eqref{eq:bound_big_phi}, we obtain
\begin{align}
\left| L'^{(2)}_k\delta^{(2)}_z - L'^{(1)}_k\delta^{(1)}_z\right|_2 &\le (\text{RHS of Eq.~}\eqref{eq:embedding_main_bound_first_part_subs}) \notag\\
&\le (\text{RHS of Eq.~}\eqref{eq:before_final_bound_first_part})\notag\\
&\le 2\varepsilon \left[d(z,y)\left\|L'^{(2)}_k\right\|_2 + \widetilde{K}\sqrt{2K(m-1)} \left|\delta^{(1)}_z\right|_2\right] + O(\varepsilon^2).
\label{eq:embedding_main_bound_first_part}
\end{align}
\end{proof}

\begin{lemma}
Let $(X,d)$ be a metric space, and let $\Delta_1 = [d^{(1)}_{rs}]$ and $\Delta_2 = [d^{(2)}_{rs}]$ be the squared distance matrices for $\{x_1,\ldots, x_m\}$ and $\{x_1,\ldots, x_{m-1},y\}$, respectively, 
where $x_1,\ldots, x_{m}, y \in X$, and $\varepsilon\equiv d(x_m,y)$. Then, we obtain
\begin{equation}
\left|\delta^{(2)}_\mu - \delta^{(1)}_\mu\right|_2 \le \frac{4 (m-1)\sqrt{K}}{m} \varepsilon + O(\varepsilon^2)
\end{equation}
for small $\varepsilon$, where $\delta^{(1)}_{\mu} = \frac{1}{m} \begin{bmatrix} \sum_{s=1}^{m} d^{(1)}_{1s}\\ \vdots\\ \sum_{s=1}^{m} d^{(1)}_{ms} \end{bmatrix}$ and $\delta^{(2)}_{\mu} = \frac{1}{m} \begin{bmatrix} \sum_{s=1}^{m} d^{(2)}_{1s}\\ \vdots\\ \sum_{r=1}^{m} d^{(2)}_{ms} \end{bmatrix}$.
\label{lemma:diff_avg_column_bound}
\end{lemma}

\begin{proof}
By combining Eqs.~\eqref{eq:define_big_phi} and \eqref{eq:perturb_dist_matrix_wtih_epsilon}, we obtain
\begin{equation}
\Delta_2 - \Delta_1 = \varepsilon\begin{bmatrix}
0 & 0 & \cdots & 0 &  \phi_1\\
0 & 0 & \cdots & 0 & \phi_2\\
\vdots & \vdots & \ddots & \vdots & \vdots\\
0 & 0 & \cdots & 0 & \phi_{m-1}\\
\phi_1 & \phi_2 & \cdots & \phi_{m-1} & 0
\end{bmatrix} + O(\varepsilon^2) = \varepsilon\Phi + O(\varepsilon^2).
\label{eq:perturbed_dist_matrix}
\end{equation}
Equation~\eqref{eq:perturbed_dist_matrix} implies that
\begin{equation}
d^{(2)}_{rs} - d^{(1)}_{rs} =
\begin{cases}
0, & r, s< m \text{ or } r=s,\\
\varepsilon\phi_r + O(\varepsilon^2), & r\in\{1,\ldots,m-1\}, s=m,\\
\varepsilon\phi_s + O(\varepsilon^2), &  r=m, s\in\{1,\ldots,m-1\}.
\end{cases}
\end{equation}
Therefore,
\begin{equation}
\delta^{(2)}_{\mu} - \delta^{(1)}_{\mu} = \frac{1}{m} \begin{bmatrix}
\sum_{s=1}^{m}\left(d^{(2)}_{1s} - d^{(1)}_{1s}\right)\\
\sum_{s=1}^{m}\left(d^{(2)}_{2s} - d^{(1)}_{2s}\right)\\
\vdots\\
\sum_{s=1}^{m}\left(d^{(2)}_{m-1, s} - d^{(1)}_{m-1, s}\right)\\
\sum_{s=1}^{m}\left(d^{(2)}_{ms} - d^{(1)}_{ms}\right)
\end{bmatrix}
= \frac{\varepsilon}{m} \begin{bmatrix}
\phi_1\\
\phi_2\\
\vdots\\
\phi_{m-1}\\
\sum_{s = 1}^{m-1} \phi_s
\end{bmatrix} + O(\varepsilon^2).
\label{eq:delta-mu-a}
\end{equation}
Using Eq.~\eqref{eq:delta-mu-a}, we obtain
\begin{align}
\left|\delta^{(2)}_\mu - \delta^{(1)}_\mu\right|_2\le \left|\delta^{(2)}_\mu - \delta^{(1)}_\mu\right|_1 &\le \frac{\varepsilon}{m} \sum_{r=1}^{m-1}\left|\phi_r\right| + \frac{\varepsilon}{m} \left| \sum_{s = 1}^{m-1} \phi_s\right| + \left|O(\varepsilon^2)\right|\notag\\
&\le \frac{\varepsilon}{m} \sum_{r=1}^{m-1}\left|\phi_r\right| + \frac{\varepsilon}{m} \sum_{s = 1}^{m-1}\left|\phi_s\right| + \left|O(\varepsilon^2)\right|\notag\\
&= 2\frac{\varepsilon}{m}\left|\left(\phi_1,\ldots,\phi_{m-1}\right)\right|_1 + O(\varepsilon^2)\notag\\
&\le 2\varepsilon\frac{\sqrt{m-1}}{m}\left|\left(\phi_1,\ldots,\phi_{m-1}\right)\right|_2 + O(\varepsilon^2)\notag\\
%
%
&= 2\varepsilon\frac{\sqrt{m-1}}{m}\sqrt{\frac{\left\|\Phi\right\|^2_F}{2}} + O(\varepsilon^2)\notag\\
&= \varepsilon\frac{\sqrt{2(m-1)}}{m}\left\|\Phi\right\|_F+ O(\varepsilon^2)\notag\\
&\le 4\varepsilon (m-1)\frac{\sqrt{K}}{m} + O(\varepsilon^2),
\label{eq:delta-mu-b}
\end{align}
where $\left|\cdot\right|_1$ denotes the $\ell^{1}$-norm, and we used Eq.~\eqref{eq:bound_big_phi} to derive the last inequality.
\end{proof}

\begin{theorem}
Let $(X,d)$ be a metric space, $\psi_1 = -\frac{1}{2}L'^{(1)}_k\left(\delta^{(1)}_z - \delta^{(1)}_\mu\right)$ be the LMDS map associated with the set of landmarks $\{x_1,\ldots, x_{m-1}, x_m\}\subseteq X$, and $\psi_2 = -\frac{1}{2}L'^{(2)}_k\left(\delta^{(2)}_z - \delta^{(2)}_\mu\right)$ be the LMDS map associated with the set of landmarks $\{x_1,\ldots, x_{m-1},y\}\subseteq X$. Furthermore, let $\lambda_1\ge\cdots\ge\lambda_k$ be the eigenvalues of $\Delta_1$ such that $\lambda_k >0$ and $\lambda_k > \lambda_{k+1}$. With any $z\in X\setminus \{x_1,\ldots,x_m, y\}$, we obtain
\begin{align}
\left|\psi_1(z) - \psi_2(z)\right|_2 &\le \varepsilon \left\{ \left\|L'^{(2)}_k\right\|_2\left[ d(z,y)+\frac{2\sqrt{K}(m-1)}{m}\right] + \widetilde{K}\sqrt{2K(m-1)}\left(\left|\delta^{(1)}_z\right|_2 + \left|\delta^{(1)}_\mu\right|_2\right)\right\}\notag\\
& + O(\varepsilon^2)
\end{align}
for small $\varepsilon \equiv d(x_m,y)$.
\label{thm:embedding_dislocation}
\end{theorem}

\begin{proof}
Let $z\in X\setminus \{x_1,\ldots,x_m, y\}$. Then,
\begin{align}
\left|\psi_1(z) - \psi_2(z)\right|_2 &= \left| -\frac{1}{2}L'^{(1)}_k \left(\delta^{(1)}_z - \delta^{(1)}_\mu \right) + \frac{1}{2}L'^{(2)}_k \left(\delta^{(2)}_z - \delta^{(2)}_\mu \right)\right|_2\notag\\
&\le \frac{1}{2}\left[\left| L'^{(2)}_k\delta^{(2)}_z - L'^{(1)}_k\delta^{(1)}_z\right|_2 + \left| L'^{(1)}_k\delta^{(1)}_\mu - L'^{(2)}_k\delta^{(2)}_\mu\right|_2 \right].
\label{eq:embedding_main_bound}
\end{align}
By Lemma~\ref{lemma:bound_first_part} (i.e., Eq.~\eqref{eq:embed_bd_first_quant}), the first term of the RHS of Eq.~\eqref{eq:embedding_main_bound} is bounded as follows:
\begin{equation}
\frac{1}{2}\left| L'^{(2)}_k\delta^{(2)}_z - L'^{(1)}_k\delta^{(1)}_z\right|_2 \le \varepsilon \left[d(z,y)\left\|L'^{(2)}_k\right\|_2+ \widetilde{K}\sqrt{2K(m-1)}\left|\delta^{(1)}_z\right|_2\right] + O(\varepsilon^2).
\label{eq:main_bound_first_part}
\end{equation}

To bound $\frac{1}{2}\left| L'^{(1)}_k\delta^{(1)}_\mu - L'^{(2)}_k\delta^{(2)}_\mu\right|_2$, we first obtain
\begin{align}
\left| L'^{(1)}_k\delta^{(1)}_\mu - L'^{(2)}_k\delta^{(2)}_\mu\right|_2 &= \left| L'^{(1)}_k\delta^{(1)}_\mu - L'^{(2)}_k\delta^{(1)}_\mu + L'^{(2)}_k\delta^{(1)}_\mu - L'^{(2)}_k\delta^{(2)}_\mu\right|_2\notag\\
&\le \left| L'^{(1)}_k\delta^{(1)}_\mu - L'^{(2)}_k\delta^{(1)}_\mu\right|_2 + \left| L'^{(2)}_k\delta^{(1)}_\mu - L'^{(2)}_k\delta^{(2)}_\mu\right|_2 \notag\\
&\le \left\| L'^{(1)}_k - L'^{(2)}_k\right\|_2 \left|\delta^{(1)}_\mu\right|_2 + \left\| L'^{(2)}_k \right\|_2 \left| \delta^{(2)}_\mu - \delta^{(1)}_\mu\right|_2\notag\\
&\le \left\| L'^{(1)}_k - L'^{(2)}_k\right\|_F \left|\delta^{(1)}_\mu\right|_2 + \left\| L'^{(2)}_k \right\|_2 \left| \delta^{(2)}_\mu - \delta^{(1)}_\mu\right|_2\notag\\
&\le \left\| L'^{(1)}_k - L'^{(2)}_k\right\|_F \left|\delta^{(1)}_\mu\right|_2 + \left\| L'^{(2)}_k \right\|_2 \cdot 4\varepsilon (m-1)\frac{\sqrt{K}}{m} + O(\varepsilon^2).
\label{eq:embedding_main_bound_second_part}
\end{align}
To derive the second and the third inequalities in Eq.~\eqref{eq:embedding_main_bound_second_part}, we used the fact that
$L'^{(1)}_k - L'^{(2)}_k$ and $L^{(2)}_k$ are bounded (see Theorem~\ref{thm:bounded_operator}).
To derive the last inequality in Eq.~\eqref{eq:embedding_main_bound_second_part}, we used Lemma~\ref{lemma:diff_avg_column_bound}. By Theorem~\ref{thm:eigen_projection_bound}, we can write $L'^{(2)}_k - L'^{(1)}_k = \varepsilon\varphi + O(\varepsilon^2)$, where $\left\|\varphi\right\|_F \le \widetilde{K}\left\|\Phi\right\|_F$. Therefore, we obtain
\begin{align}
\left| L'^{(1)}_k\delta^{(1)}_\mu - L'^{(2)}_k\delta^{(2)}_\mu\right|_2 \le
(\text{RHS of Eq.~}\eqref{eq:embedding_main_bound_second_part})
& \le \varepsilon\widetilde{K}\left\|\Phi\right\|_F\left|\delta^{(1)}_\mu\right|_2 +  \left\| L'^{(2)}_k \right\|_2 \cdot 4 \varepsilon (m-1)\frac{\sqrt{K}}{m} + O(\varepsilon^2) \notag\\
&\le 2\varepsilon\sqrt{K} \left[\widetilde{K}\sqrt{2(m-1)}\left|\delta^{(1)}_\mu\right|_2 + \left\| L'^{(2)}_k \right\|_2 \frac{2(m-1)}{m}\right] \notag\\
& + O(\varepsilon^2),
\label{eq:before_final_bound_second_part}
\end{align}
where we used Eq~\eqref{eq:bound_big_phi} to derive the last inequality in Eq.~\eqref{eq:before_final_bound_second_part}. 

By substituting Eqs.~\eqref{eq:main_bound_first_part} and \eqref{eq:before_final_bound_second_part} in Eq.~\eqref{eq:embedding_main_bound}, we obtain
\begin{align}
& \left|\psi_1(z) - \psi_2(z)\right|_2 \notag\\
%
%
\le& \varepsilon \left[d(z,y)\left\|L'^{(2)}_k\right\|_2+ \widetilde{K}\sqrt{2K(m-1)}\left|\delta^{(1)}_z\right|_2 + \frac{1}{2} \cdot 2\varepsilon\sqrt{K} \left[\widetilde{K}\sqrt{2m-2}\left|\delta^{(1)}_\mu\right|_2 + \left\| L'^{(2)}_k \right\|_2 \frac{2(m-1)}{m}\right]\right] \notag\\
& + O(\varepsilon^2) \notag\\
=& \varepsilon \left\{ \left\|L'^{(2)}_k\right\|_2\left[ d(z,y)+\frac{2\sqrt{K}(m-1)}{m}\right] + \widetilde{K}\sqrt{2K(m-1)}\left(\left|\delta^{(1)}_z\right|_2 + \left|\delta^{(1)}_\mu\right|_2\right)\right\} + O(\varepsilon^2),
\end{align}
as desired.
\end{proof}

In sum, our main result in this section is Theorem~\ref{thm:embedding_dislocation}, which upper bounds the size of change in the embedding coordinate of any data point when just one landmark is replaced. We found that the coordinate change is small when the new landmark and the dropped landmark are close.  Theorem~\ref{thm:embedding_dislocation} is applicable to the landmark replacement implemented by Algorithm~\ref{alg:landmark_replacement}.

\section{Numerical results}

In this section, we demonstrate the proposed landmark replacement algorithm with three time series data sets,
one of which is synthetic and the other two are empirical. For each data set, we use the LMDS as the out-of-sampling embedding method to map the time series data into the two-dimensional Euclidean space. For each data set, we consider three embedding scenarios. In the first scenario, referred to as offline embedding with initial landmarks, we use the first $m$ data points as landmarks and never replace them as newer data points arrive. In the second scenario, referred to as online embedding, we update landmarks using Algorithm~\ref{alg:landmark_replacement} as each new data point arrives. We initialize the geometric graph for Algorithm~\ref{alg:landmark_replacement} with $\rho = 10^{-20}$, with which the first $m$ nodes are isolated in the initial geometric graph. In the third scenario, referred to as offline embedding with random landmarks, after all the points have been revealed, we select $m$ data points uniformly at random as landmarks and embed the entire data set. We note that the first and third scenarios are offline out-of-sample embedding methods under different conditions on the available information. 

\subsection{S-curve}

The S-curve
%
%
is a synthetic data set, constructed by the sklearn package in python~\cite{sklearn_api}. We generate the S-curve in $\mathbb{R}^3$ with $10^3$ points, which we visualize in Fig.~\ref{fig:s_curve}. Because the S-curve is specified by the Cartesian coordinate system of the two-dimensional S shape on the $(x, z)$ plane and $0 \le y \le 2$, we linearly ordered the $10^3$ points according to the $(x, z)$ coordinate so that the point nearest to the right-top corner of the S shape in the $(x, z)$ plane is the first, and that nearest to the left-bottom corner of the S shape is the last. The color code in the figure is a guide to the eyes to indicate the ordering of the points.


The black diamonds in Fig.~\ref{fig:s_curve} show the final landmarks in each of the three embedding scenarios. We set $m=100$; we verified the robustness of the online embedding with respect to the value of $m$ in electronic supplementary material section~\ref{sec:dependence_on_m}, including for the other two data sets analyzed in the following sections. Note that, in the offline embedding with initial landmarks (shown in panel (a)) and the offline embedding with random landmarks (shown in panel (c)), the landmarks do not change over time. Figure~\ref{fig:s_curve}(b) indicates that, with the online embedding, many final landmarks are devoted to early data points while some newer landmarks are scattered over the middle and lower parts of the S-curve, covering various parts of the S-curve. The final $\rho = 0.94$ for the online embedding.

We show the results of the LMDS into $\mathbb{R}^2$ in Fig.~\ref{fig:s_curve_embedding} for the three embedding scenarios. Figure~\ref{fig:s_curve_embedding}(a) shows that the shape of the S-curve is not preserved in the offline embedding with initial landmarks, presumably because the set of the first $m$ data points is only good at capturing local geometry of the S-curve, i.e., the beginning part of the S-curve. In contrast, Figs.~\ref{fig:s_curve_embedding}(b) and \ref{fig:s_curve_embedding}(c) preserve the structure of the original S-curve well. To quantitatively assess the quality of embedding, we show in Table~\ref{table:goodness_of_embedding} the values of $\sigma_L$ and $\sigma$ for each scenario; see the rows with $k=2$, where we remind that $k$ is the embedding dimension. For the offline embedding with random landmarks, we show the average and standard deviation of $\sigma_L$ and $\sigma$ calculated based on 100 runs. The normalized stress for the entire data, $\sigma$, is substantially larger for the offline embedding with initial landmarks (i.e., $\sigma = 0.23$) than the online embedding (i.e., $\sigma = 0.16$) and the offline embedding with random landmarks (i.e., $\sigma = 0.15 \pm 0.004$, and $\sigma=0.15$ for the run shown in Figs~\ref{fig:s_curve}(c) and \ref{fig:s_curve_embedding}(c)). The latter two embedding methods yield similar $\sigma$ values. These results support that our online algorithm performs on par with the offline embedding with random landmarks, which requires the knowledge of all the data points beforehand, and better than the offline embedding with initial landmarks. 
In terms of $\sigma_L$, the offline embedding with initial landmarks is better than the other two embedding algorithms. This is because the former has all the $m$ points geometrically close to each other on the approximately planar part of the S-curve, which renders embedding of the $m$ landmarks an easy task.

To justify the choice of embedding dimension, $k=2$, we also calculated $\sigma$ and $\sigma_L$ for the three embedding scenarios for $k=1$ and $k=3$.
We show the results in Table~\ref{table:goodness_of_embedding}. We find that $\sigma$ (as well as $\sigma_L$) is much larger than an acceptable range ($\sigma \le 0.15$~\cite{borg2005modern}) for any of the three embedding schemes with $k=1$. Therefore, using $k=1$ is not appropriate. In contrast, all the three embedding scenarios yield $\sigma, \sigma_L < 0.01$ when $k=3$. In other words, there is little distortion by LDMS, which is because the original S-curve is defined in the three-dimensional Euclidean space.
Therefore, we consider that embedding dimension $k=2$ is an adequate challenge for approximate distance-preserving embedding of the S-curve data.

Last, we further show in Table~\ref{table:goodness_of_embedding} the average and the standard deviation of $\sigma_L$ and $\sigma$ for the random online embedding. (See electronic supplementary material section~\ref{sec:rand_online} for the definition of the algorithm.) While the value of $\sigma$ for random online embedding is on par with that of our main online embedding, the value of $\sigma_L$ is substantially larger for the former than the latter.

\begin{figure}[t]
    \centering
     \includegraphics[width=\textwidth]{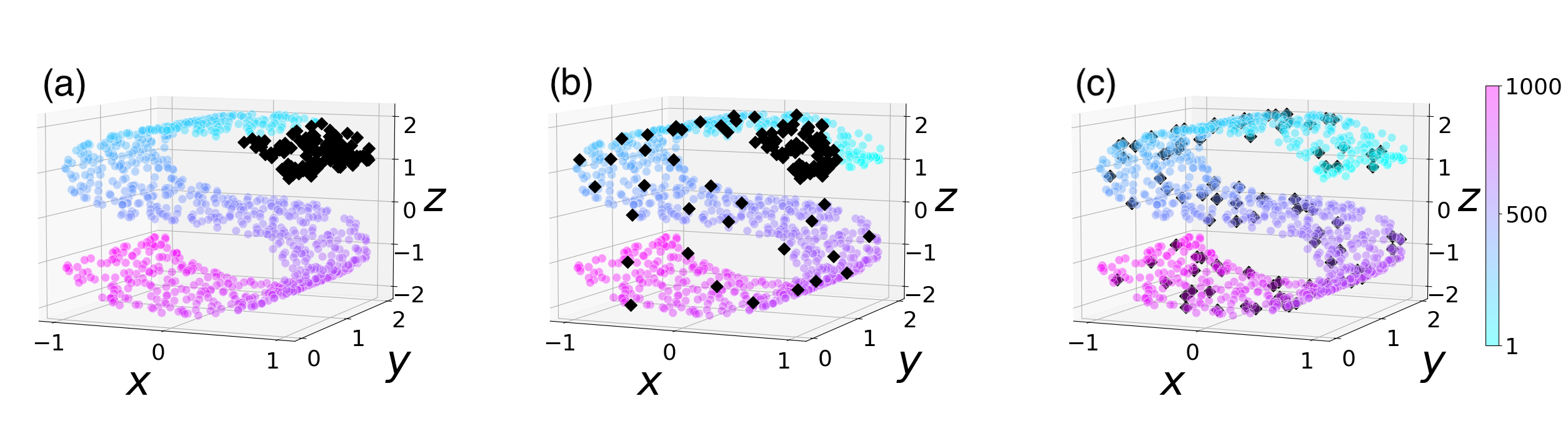}
     \caption{Three-dimensional S-curves composed of $10^3$ points, shown by the circles. (a) Offline embedding with initial landmarks. (b) Online embedding using Algorithm~\ref{alg:landmark_replacement}. (c) Offline embedding with random landmarks.  The color indicates the artificially defined order in which each data point arrives. Each diamond represents a landmark. In (b), the diamonds are the final landmarks. We set $m = 100$.}
\label{fig:s_curve}
\end{figure}

\begin{figure}[t]
    \centering
     \includegraphics[width=\textwidth]{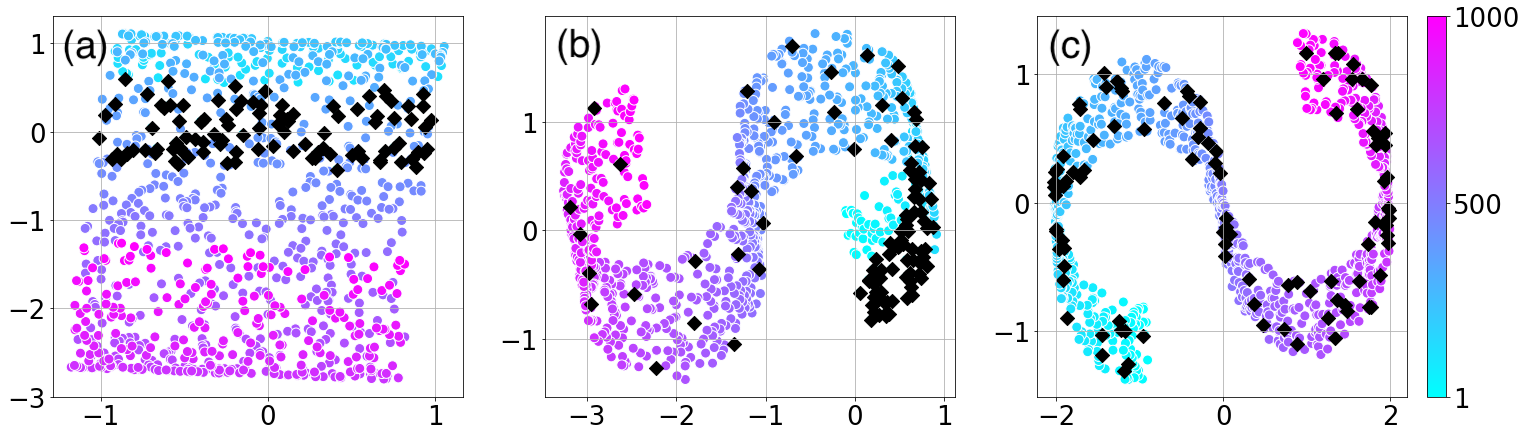}
     \caption{Two-dimensional embedding of the S-curve shown in Fig.~\ref{fig:s_curve}. (a) Offline embedding with initial landmarks. (b) Online embedding using Algorithm~\ref{alg:landmark_replacement}. (c) Offline embedding with random landmarks. Each diamond represents a landmark. In (b), the diamonds are the final landmarks.}
\label{fig:s_curve_embedding}
\end{figure}

\begin{table}[t]
\begin{center}
\caption{Goodness of the LMDS for the three data sets. Columns ``Initial'' and ``Random'' represent the offline embedding with the initial and random landmarks, respectively. ``Random online'' refers to the random online landmark replacement algorithm described in electronic supplementary material section~\ref{sec:rand_online}. In the last two columns, we report the average $\pm$ standard deviation on the basis of 100 runs.}
\label{table:goodness_of_embedding}
\begin{tabular}{|c|c|c|c|>{\centering\arraybackslash}p{1.3cm}|>{\centering\arraybackslash}p{1.3cm}|>{\centering\arraybackslash}p{2.3cm}|>{\centering\arraybackslash}p{2.3cm}|}
 \hline
Data & $m$ & $k$ & Index & Initial & Online & Random & Random online \\
\hline
\multirow{6}{*}{S-curve} & \multirow{6}{*}{100}
&\multirow{2}{*}1
& \multirow{1}{*}{$\sigma_{L}$}  & $0.21$  & 0.40 & $0.31\pm 0.01$ & $0.51\pm 0.03$ \\
\cline{4-8}
& & & $\sigma$  & 0.74  & 0.33 & $0.31\pm 0.002$ & $0.38\pm 0.10$ \\
\cline{3-8}
& &\multirow{2}{*}2
& \multirow{1}{*}{$\sigma_{L}$}  & $<0.01$  & 0.15 & $0.14\pm 0.01$ & $0.25\pm 0.02$\\
\cline{4-8}
& & & $\sigma$  & 0.23  & 0.16 & $0.15\pm 0.004$ & $0.17\pm 0.06$ \\
\cline{3-8}
\cline{3-8}
& &\multirow{2}{*}3
& \multirow{1}{*}{$\sigma_{L}$}  & $<0.01$  & $<0.01$ & $<0.01$ & $<0.01$ \\
\cline{4-8}
& & & $\sigma$  & $<0.01$  & $<0.01$ & $<0.01$ & $<0.01$ \\
\cline{3-8}
\hline

\multirow{6}{*}{SPY} & \multirow{6}{*}{10}
&\multirow{2}{*}1
& \multirow{1}{*}{$\sigma_{L}$}  & $0.29$  & 0.15 & $0.07\pm 0.03$ & $0.22\pm 0.08$ \\
\cline{4-8}
& & & $\sigma$  & $0.80$  & 0.09 & $0.09\pm 0.01$ & $0.41\pm 0.16$\\
\cline{3-8}
& &\multirow{2}{*}2
& \multirow{1}{*}{$\sigma_{L}$}  & 0.14  & 0.05 & $0.03\pm 0.01$ & $0.08\pm 0.02$ \\
\cline{4-8}
& & & $\sigma$  & 0.68  & 0.04 & $0.06\pm 0.02$ & $0.21\pm 0.06$ \\
\cline{3-8}
& &\multirow{2}{*}3
& \multirow{1}{*}{$\sigma_{L}$}  & 0.07  & 0.02 & $0.01\pm 0.01$ & $0.04\pm 0.01$ \\
\cline{4-8}
& & & $\sigma$  & 0.43  & 0.02 & $0.05\pm 0.02$ & $0.15\pm 0.05$ \\
\cline{3-8}
\hline

\multirow{6}{*}{Hospital} & \multirow{6}{*}{20}
& \multirow{2}{*}1
& \multirow{1}{*}{$\sigma_{L}$} & 0.13 & 0.16 & $0.21\pm 0.04$ & $0.28\pm 0.05$ \\
\cline{4-8}
& & & $\sigma$ & 0.61 & 0.22 & $0.23\pm 0.01$ & $0.24\pm 0.02$ \\
\cline{3-8}
& & \multirow{2}{*}2
& \multirow{1}{*}{$\sigma_{L}$}  & $<0.01$  & 0.06 & $0.08\pm 0.02$ & $0.11\pm 0.02$ \\
\cline{4-8}
& & & $\sigma$  & 0.40  & 0.11 & $0.11\pm 0.01$ & $0.11\pm 0.02$ \\
\cline{3-8}
& & \multirow{2}{*}3
& \multirow{1}{*}{$\sigma_{L}$}  & $<0.01$  & $<0.01$ & $0.04\pm 0.01$ & $0.03\pm 0.01$\\
\cline{4-8}
& & & $\sigma$  & 0.29  & 0.07 & $0.06\pm 0.01$ & $0.07\pm 0.01$ \\
\cline{3-8}
\hline
\end{tabular}
\end{center}
\end{table}

\subsection{Stock price time series}

Our second data set is the daily adjusted closing prices of the top ten holding stocks, in terms of market capitalization, from SPDR S\&P 500 ETF Trust (SPY). SPY is a member of Standard \& Poor's depository receipt (SPDR), which is a family of exchange-traded funds (ETFs). We have normalized the adjusted closing prices of each stock by the map $x\mapsto (x - x_{\min})/(x_{\max} - x_{\min})$, where $x$ is the adjusted closing price of a stock on a valid day, and $x_{\min}$ and $x_{\max}$ are the lowest and the highest prices of the stock in the observation period, respectively. As of June 2023, SPY is one of the largest ETFs by market capitalization. The SPY's ten largest stock holdings are Apple Inc.\,(AAPL), Microsoft Corp.\,(MSFT), Amazon (AMZN), NVIDIA Corp.\,(NVDA), Alphabet Inc.\,Class A (GOOGL), Tesla Inc.\,(TSLA), Alphabet Inc. Class C (GOOG), Berkshire Hathaway Inc.\,Class B (BRK.B), Meta Platforms Inc.\,Class A (META), and UnitedHealth Group Inc.\,(UNH). We collect the daily adjusted closing prices of these ten stocks from Yahoo finance from 1/2/2013 to 30/12/2022, resulting in $n = 2,518$ valid days in total.

Each data point, $x_{r}$, where $r\in \{1, \ldots, n\}$, is a ten-dimensional vector, whose each component is each stock's price on the $r$th valid day.
We use the Euclidean distance to measure the distance between each pair of data points. Then, we embed $\{x_1, \ldots, x_n\}$ into the two-dimensional Euclidean space. We set the number of landmarks $m=10$. The final $\rho = 0.94$ for the online embedding.

\begin{figure}[t]
    \centering
     \includegraphics[width=\textwidth]{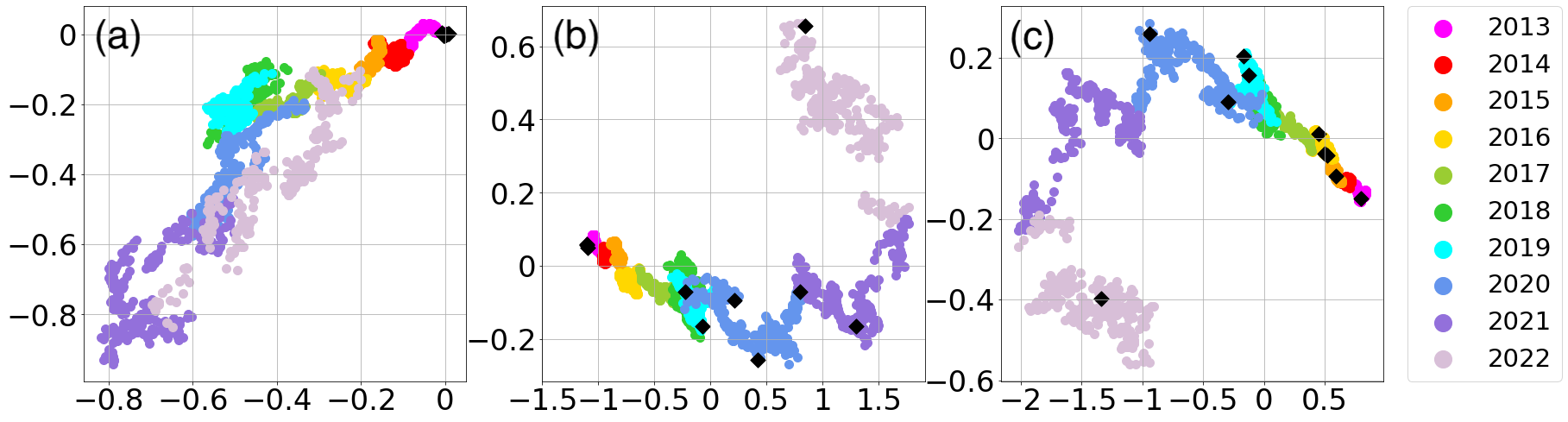}
     \caption{Two-dimensional embedding of the time series of the ten largest stocks of SPY. Each circle represents the data point on a valid day. The color of the circle represents the year. Each diamond represents a (final) landmark. (a) Offline embedding with initial landmarks. (b) Online embedding. (c) Offline embedding with random landmarks. We set $m=10$.}
\label{fig:etf_embedding}
\end{figure}

\begin{figure}[t]
    \centering
     \includegraphics[width=\textwidth]{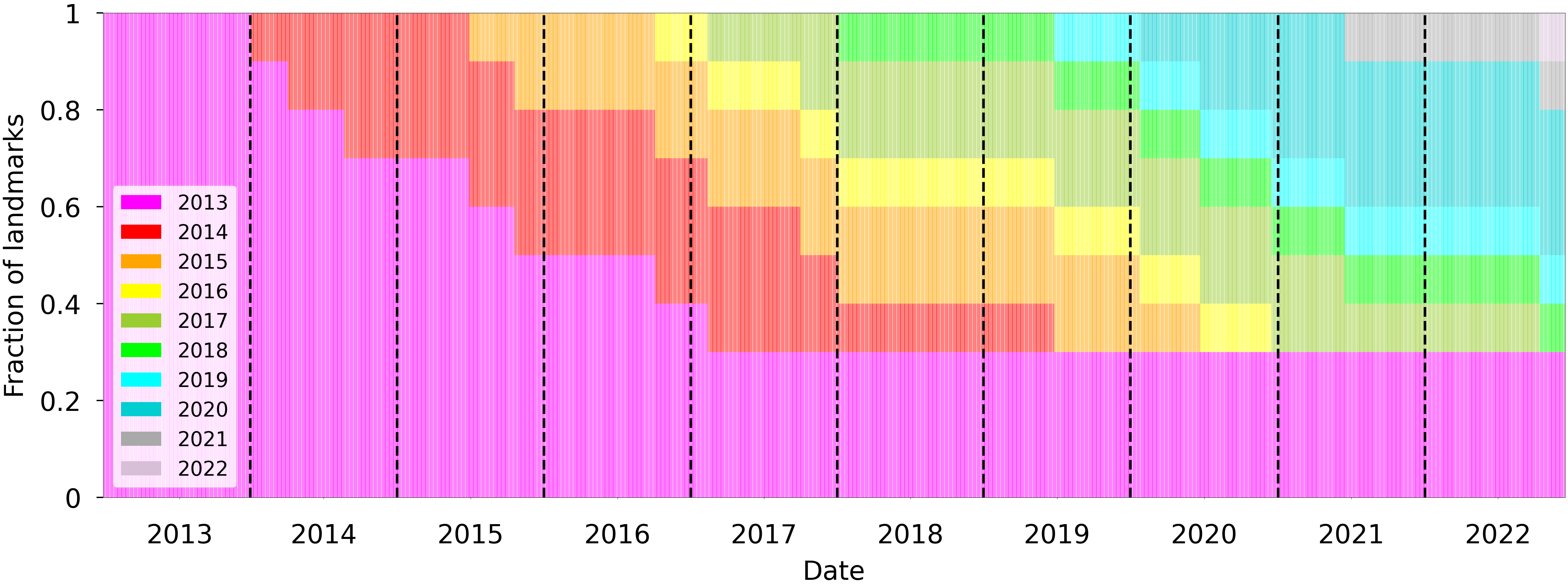}
     \caption{Fraction of landmarks by year when we embed the time series of the price of the ten largest stocks of SPY using the online embedding. The time span shown is from 1/15/2013 to 12/30/2022, i.e., from the 11th to the 2,518th days. The dashed lines represent the boundary between years.}
\label{fig:etf_fraction_of_landmarks}
\end{figure}

We show in Fig.~\ref{fig:etf_embedding} the embedding coordinates of the daily adjusted closing prices of ten stocks. Unlike in Fig.~\ref{fig:s_curve}(b), which shows that a majority of the final landmarks are data points arriving in early time steps (shown by the diamonds), seven out of the ten final landmarks shown in Fig.~\ref{fig:etf_embedding}(b) are data points arriving in 2018 or later, and the remaining three landmarks are in 2013. To examine evolution of landmarks over years, we show in Fig.~\ref{fig:etf_fraction_of_landmarks} the fraction of landmarks from each year, measured on each day starting on the 11th ($= (m+1)$th) day of the data stream. Specifically, to calculate the fraction of the year of the landmarks, we count the number of landmarks that correspond to each year and divide the number by $m$ $(=10)$. For example, if the fraction of year 2013 is 0.2 at time $t$, it means that two landmarks at time $t$ are dated year 2013. Figure~\ref{fig:etf_fraction_of_landmarks} indicates that our online embedding algorithm selects fewer data points in 2016, 2018, and 2019 as landmarks, even transiently. An intuitive explanation of this phenomenon is that the data points in these years do not explore a new part of the embedding space (see Fig.~\ref{fig:etf_embedding}(b)), or presumably in the original ten-dimensional space, either.

The trajectory in the embedding space obtained with the offline embedding with initial landmarks looks qualitatively different from those obtained with the online embedding and the offline embedding with random landmarks. Specifically, the data points in 2022 heavily overlap with those between 2016 and 2021 in the case of the offline embedding with initial landmarks (see Fig.~\ref{fig:etf_embedding}(a)). In contrast, with the online embedding and the offline embedding with random landmarks, the data points in 2022 stretch into part of the embedding space that has been unexplored by the trajectory in the prior years (see Figs.~\ref{fig:etf_embedding}(b) and \ref{fig:etf_embedding}(c)).
Similarly to the case of the S-curve data set, we hypothesize that the set of the first $m$ data points, which are the landmarks in the offline embedding with initial landmarks, only captures local geometry of the data set near the first $m$ data points and fails to capture the structure of the later data points that are far from the first $m$ landmarks. To discuss this possibility, we show $\sigma$ and $\sigma_L$ for the three embedding scenarios in Table~\ref{table:goodness_of_embedding}. The table indicates that the online embedding and the offline embedding with random landmarks perform better than the offline embedding with initial landmarks in terms of both $\sigma$ and $\sigma_L$ for this data set, affirming our hypothesis. The table also indicates that the online embedding and the offline embedding with random landmarks produce almost the same quality of embedding (i.e., $\sigma = 0.04$). 

We show in Table~\ref{table:goodness_of_embedding} the values of $\sigma$ and $\sigma_L$ for embedding dimensions $k=1,3$ and for each of the three embedding scenarios. The offline embedding with initial landmarks gives large $\sigma$ values, implying poor embedding, for both $k=1$ and $k=3$, as expected. It gives an acceptable $\sigma_L$ value for $k=3$, which only suggests successful embedding of the landmarks, not the non-landmark points. With the online embedding and offline embedding with random landmarks, $\sigma$ is sufficiently small (i.e., $\le 0.1$) for both $k=1$ and $k=3$ as well as for $k=2$. However, the gap in terms of the $\sigma$ value is larger between $k=1$ and $k=2$ than between $k=2$ and $k=3$, and the same holds true for $\sigma_L$. We chose $k=2$ for these reasons.

Last, for this data set, the random online embedding performs worse than our online embedding algorithm in terms of both $\sigma_L$ and $\sigma$ (see Table~\ref{table:goodness_of_embedding}).

\subsection{Hospital}

As a third data set, we analyze a temporal network data set using the LMDS-based embedding method proposed in our previous study~\cite{thongprayoon2023embedding}. The data set, which we refer to as the Hospital data set, is a sequence of time-stamped contact events between two individuals in a hospital ward in Lyon, France, recorded by the SocioPatterns project~\cite{vanhems2013estimating}. We transform the data set into a continuous-time temporal networks, i.e., weighted adjacency matrix as a function of continuous time, using the tie-decay network model~\cite{ahmad2021tie}.
The constructed temporal network contains 75 nodes, of which 11 nodes represent medical doctors, 35 nurses, and 29 patients. There are $1,139$ edges and $32,424$ contacts in total. The data set was recorded between 1:00 PM, 12/6/2010, Monday and 2:00 PM, 12/10/2010, Friday. There are $n = 9,453$ time points at which there is at least one contact event, which we use as $t_1$, $\ldots$, $t_n$.

We review tie-decay networks and the embedding method based on them~\cite{thongprayoon2023embedding} in electronic supplementary material section~\ref{sec:tie_decay_embedding}. The tie-decay network is a $75\times 75$ matrix that varies over time, and each of its entry represents the time-dependent strength of the tie between a pair of nodes. The time being continuous implies that the network is not limited to the times at which at least one contact event occurs. We embed the tie-decay network obtained from the Hospital data into the two-dimensional Euclidean space. The embedding method is based on the LMDS and therefore requires a distance measure between pairs of data points, which are networks in the present case. We employ the unnormalized Laplacian distance as network distance measure because it yielded satisfactory embedding in our previous study~\cite{thongprayoon2023embedding}. We set $m =20$ and the decay rate of the tie strength $\alpha = 10^{-2}$ (see electronic supplementary material section~\ref{sec:tie_decay_embedding} for the definition of $\alpha$). The final $\rho = 22.79$ for the online embedding.

We show in Fig.~\ref{fig:hosp_embedding} the hourly averages of the embedding coordinates for each of the three embedding scenarios. For example, the coordinate labeled 2 PM represents the averaged coordinate of the trajectory between 2:00 PM and 3:00 PM. Similar to the case of the S-curve and SPY data sets, the offline embedding with initial landmarks produces an embedding trajectory that is qualitatively different from the online embedding and the offline embedding with random landmarks. Specifically, the embedded data points are roughly linearly distributed in the two-dimensional embedding space (see Fig.~\ref{fig:hosp_embedding}(a)). In contrast, the trajectory is U-shaped in the case of the online embedding (see Fig.~\ref{fig:hosp_embedding}(b)). The shape of the trajectory is roughly intermediate between the trajectories for these two cases for the offline embedding with random landmarks (see Fig.~\ref{fig:hosp_embedding}(c)).

To illustrate how the landmarks are distributed over time and days, we show in Fig.~\ref{fig:hosp_landmark_dist} the time distribution of the $m = 20$ landmarks for each of the three embedding scenarios. Each diamond represents a landmark. By definition, the landmarks for the offline embedding with initial landmarks are concentrated around 1:00 PM, Monday, specifically, between 1:00 PM to 1:12 PM (see Fig.~\ref{fig:hosp_landmark_dist}(a)). The landmarks for the offline embedding with random landmarks are widely distributed over the different days and time (Fig.~\ref{fig:hosp_landmark_dist}(c)), which is also expected. In contrast, Fig.~\ref{fig:hosp_landmark_dist}(b) indicates that all but two final landmarks for the online embedding are the networks between 1:00 PM and 2:00 PM, Monday, and the other two landmarks are from Monday at 2:14 PM and Tuesday at 10:16 AM. We interpret that this result occurs because the embedding trajectory does not much explore new part of the embedding space after Monday such that the landmarks from Monday are mostly sufficient.

Table~\ref{table:goodness_of_embedding} indicates that the offline embedding with initial landmark yields $\sigma = 0.40$ (see the rows with $k=2$), which is substantially larger than the $\sigma$ values for the other two embedding scenarios. Therefore, the offline embedding with initial landmarks is not satisfactory despite that it is good at capturing the local geometry of the set of the initial landmarks (with $\sigma_L < 0.01$). We obtain $\sigma=0.11, 0.11\pm 0.01$
for the online embedding and the average of the offline embedding with random landmarks, respectively. These $\sigma$ values are roughly as small as that when we use all the 9,453 data points at which any contact event arrives as landmarks (i.e., $\sigma = 0.09$~\cite{thongprayoon2023embedding}). Therefore, we conclude that our online embedding algorithm with a small number of landmarks (i.e., $m=20$) and the two-dimensional embedding space is fairly satisfactory for this data set.

We previously showed for this data set that two networks at the same time of different days tend to be more similar than two networks whose times of the day are different \cite{thongprayoon2023embedding}. To examine this property for our online landmark replacement algorithm,
%
%
%
we compute the distances between arbitrary pairs of hours in each of Figs.~\ref{fig:hosp_embedding}(a), \ref{fig:hosp_embedding}(b), and \ref{fig:hosp_embedding}(c). For each embedding scenario, we group the distance values into two sets. The first set is composed of the distance values for an arbitrary pair of networks at the same time of the day. The second set is composed of the distance values for an arbitrary pair of networks at different times of the day. 
If the trajectory is similar across the different days, then the distance for the first set will be smaller than that for the second set.
Therefore, we run the Student's t-test for unequal sample sizes to test the difference in the average distance between the two sets of distance values. The $p$-values for the offline embedding with initial landmarks, online embedding, and offline embedding with random landmarks are equal to $2.46\times 10^{-9}, 8.08\times 10^{-15}$, and $5.74\times 10^{-15}$, respectively. Therefore, the Hospital temporal contact network produces similar trajectories over the different days, including the time, according to the t-test, no matter which of the three embedding scenarios we adopt. However, the $p$-value for the offline embedding with initial landmarks is roughly $10^{6}$ times larger than those for the online embedding and the offline embedding with random landmarks. Furthermore, the two-dimensional embedding in which we use all the $n = 9,453$ time points as landmarks yields $p = 3.67\times 10^{-15}$, which is roughly $10^6$ times smaller than that for the offline embedding with initial landmarks, but only roughly $2.2$ and $1.6$ times smaller than those for the online embedding and the offline embedding with random landmarks, respectively. This result lends another support to the strength of our online embedding algorithm.

Next, we highlight that Table~\ref{table:goodness_of_embedding} also provides the $\sigma$ and $\sigma_L$ values for embedding dimensions $k=1$ and $k=3$. Regardless of the $k$ ($\in \{ 1, 2, 3 \}$) value, $\sigma$ for the offline embedding with initial landmarks is unacceptably large and larger than $\sigma$ for
the online embedding and the offline embedding with random landmarks. For the online embedding and the offline embedding with random landmarks with $k=1$,
we obtain $\sigma = 0.22$ and $\sigma = 0.23\pm 0.01$, respectively. These $\sigma$ values are higher than a guideline cut-off value of $0.15$~\cite{borg2005modern}. In contrast, $\sigma$ is small enough (i.e., $\sigma \le 0.11$) for these two embedding algorithms when $k=2$ and $k=3$. The gap in $\sigma$ is larger between $k=1$ and $k=2$ than between $k=2$ and $k=3$, which is why we have mainly used $k=2$. These results on the dependence of $\sigma$ on $k$ are similar to those for the SPY data.

Last, the random online embedding performs worse than or on par with our online embedding algorithm in terms of $\sigma_L$ or $\sigma$, respectively, for the present data set (see Table~\ref{table:goodness_of_embedding}).

\begin{figure}[t]
    \centering
     \includegraphics[width=\textwidth]{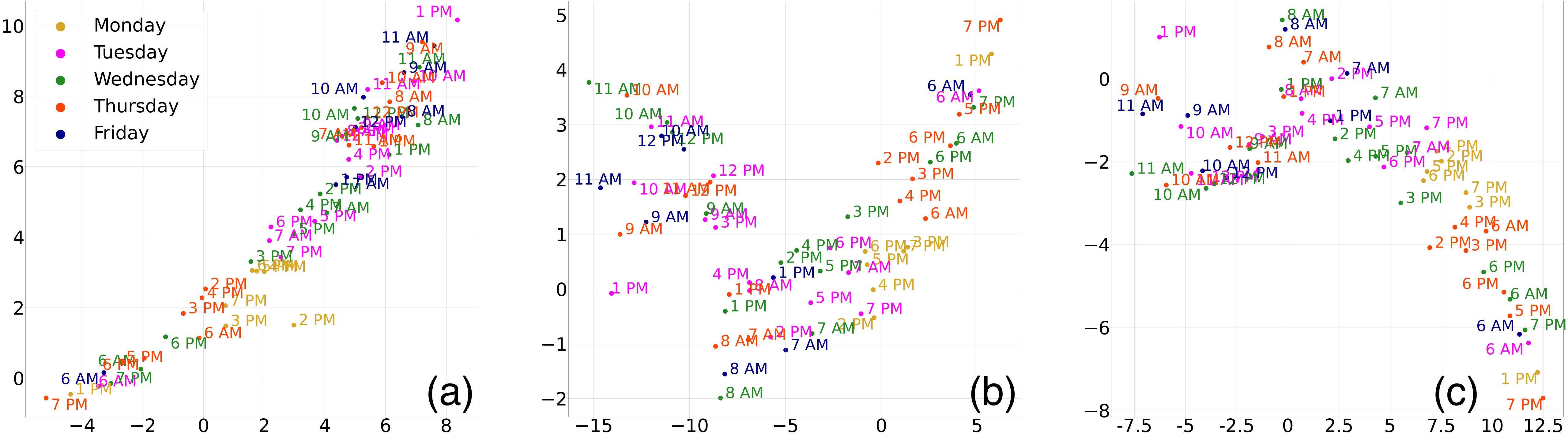}
     \caption{Hourly average embedding coordinate of the Hospital temporal network across weekdays. (a) Offline embedding with initial landmarks. (b) Online embedding. (c) Offline embedding with random landmarks. We set $m = 20$ and $\alpha = 10^{-2}$.}
\label{fig:hosp_embedding}
\end{figure}

\begin{figure}[t]
    \centering
     \includegraphics[width=\textwidth]{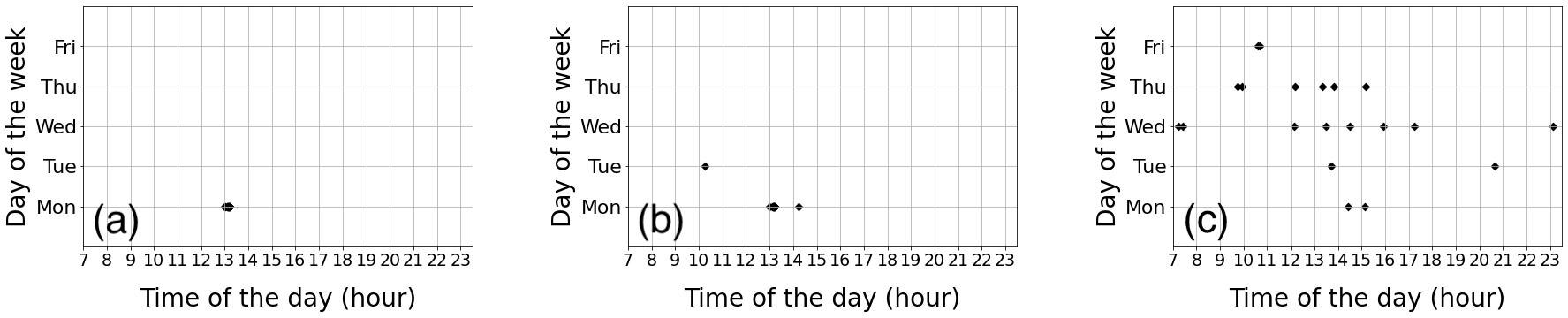}
     \caption{Distribution of the final landmarks over time for the Hospital temporal network. The midnight corresponds to $0$ and $24$ on the horizontal axis. (a) Offline embedding with initial landmarks. (b) Online embedding. (c) Offline embedding with random landmarks.}
\label{fig:hosp_landmark_dist}
\end{figure}

\subsection{Computation time}

In this section, we numerically examine the computation time and time courses of $\rho$ for our online embedding algorithm. We set the embedding dimension $k=2$. We carried out the computations on a Macbook pro with an Apple M1 chip as the processing unit and 8 GB of RAM.
See electronic supplementary material section~\ref{sec:hosp_trajectory_induced_jump} for the time courses of $\rho$.

In practice, a landmark replacement does not occur every time a new data point, $x_{r}$, arrives. Therefore, the actual computation time may be smaller than the theoretical results derived in Corollary~\ref{corr:time_complexity_sequential_arrival}, i.e., $O(n^3)$, which corresponds to the worst case. We show in Figs.~\ref{fig:time_complexity}(a), \ref{fig:time_complexity}(b), and \ref{fig:time_complexity}(c) the time consumed for sequentially embedding the first $n$ data points in the S-curve, SPY, and Hospital data sets, respectively. Note that we used $n$ to denote the number of all data points in the previous sections, whereas here we use $n$ to denote an arbitrary number of data points. However, we believe that there is no confusion.
 
Figures~\ref{fig:time_complexity}(a) and \ref{fig:time_complexity}(b) suggest that the actual computation time approximately scales as $O(n^2)$ rather than $O(n^3)$. In the figure, we show $\propto n^2$ and $\propto n^3$ lines as a guide to the eyes, where $\propto$ represents ``in propotion to''.  
We need to be careful when assessing the computation time for the Hospital data because it is a temporal network, whereas the S-curve and SPY data are
conventional multidimensional time series. We recall that we have used the unnormalized Laplacian network distance measure for the Hospital data and the Euclidean distance for the S-curve and SPY data. Computing the former takes much longer time, specifically, 
$O(n^3)$~\cite{thongprayoon2023embedding}, than computing the latter because the former requires computation of eigenvalues. Therefore, we precompute the Laplacian distance for each pair of static networks in the tie-decay temporal network. Then, we carry out our online embedding algorithm by referring to the precomputed (and stored) Laplacian distance values.
We show in Fig.~\ref{fig:time_complexity}(c) two types of computation time for embedding the Hospital data.
One is the total computation time. The other is the computation time excluding the calculation of the Laplacian distance between each pair of static networks.
We find that the computation time without the calculation of the Laplacian distance roughly scales as $O(n^3)$. Therefore, the calculation of the Laplacian distance is not a bottleneck of computation. The total computation time starts at a relatively large value but seems to approach the $O(n^3)$ scaling as $n$ increases. In conclusion, the numerically evaluated time complexity is at most cubic and smaller for the S-curve and SPY data. These results are consistent with our theoretical results.

\begin{figure}[t]
    \centering
     \includegraphics[width=\textwidth]{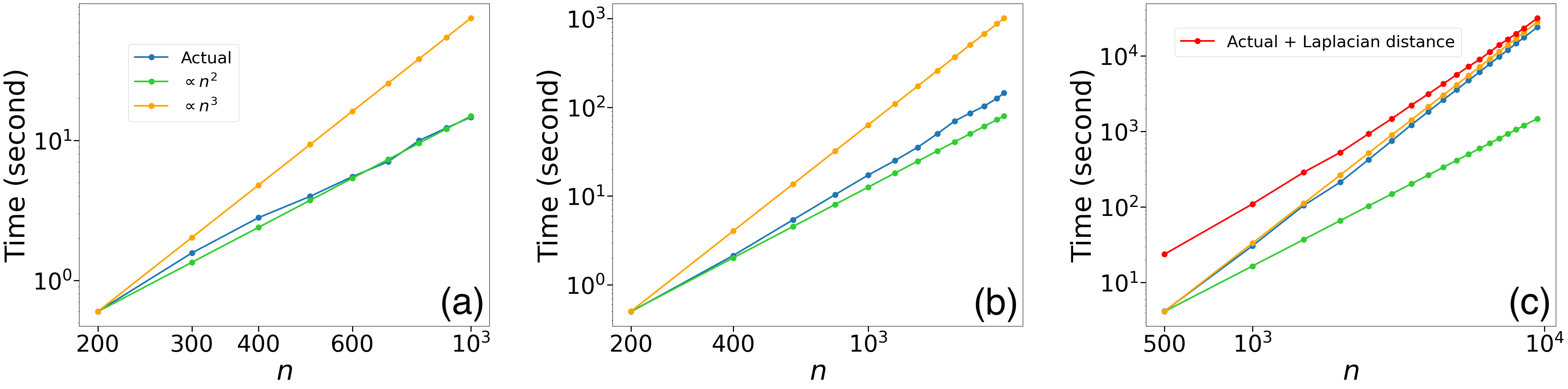}
     \caption{Numerically evaluated computation time of the proposed online embedding algorithm. (a) S-curve. (b) SPY. (c) Hospital. The green and orange straight lines represent $\propto n^2$ and $\propto n^3$, respectively. The intercept of these lines is arbitrary. In (c), we show the total computation time and the computation time excluding the computation of the Laplacian distance matrix by different lines.}
\label{fig:time_complexity}
\end{figure}

\section{Discussion}

We developed an algorithm of replacing landmarks in an online fashion for out-of-sample embedding techniques. In particular, we employed the framework of geometric graphs and dominating sets to capture geometry of landmarks and other data points. We then provided mathematical underpinning of the proposed algorithm on the computation time and perturbation of the embedding coordinate upon a landmark replacement event. We finally carried out numerical simulations to verify the performance of the algorithm combined with the LMDS. On the three data sets examined, an embedding dimension of two was practically enough, and the algorithm performed on par with the offline counterpart of the LMDS in which the fixed landmarks are selected uniformly at random from all the data points (which is by definition impossible for any online algorithm).

Our online landmark replacement scheme crucially depends on variable $\rho$, including its initial value. Variable $\rho$ regulates the number of edges in the geometric graph used in Algorithm~\ref{alg:landmark_replacement}. If the arrival of a new data point triggers an actual landmark replacement event, then $\rho$ increases. 
We numerically verified that $\rho$ drastically increases when the trajectory explores new part of the data space and that $\rho$ saturates or only gradually increases after some time
(see electronic supplementary material section~\ref{sec:hosp_trajectory_induced_jump}).
Our algorithm forces $\rho$ to grow monotonically in time and lacks a mechanism to limit and shrink the size of $\rho$. A large $\rho$ implies that the update of the geometric graph is insensitive to the position of new data points. For example, if $\rho$ is extremely large, then any new data point would be connected to all the existing nodes. We suspect that $\rho$ in later stages of our simulations is larger than would be necessary such that the minimal dominating set of the geometric graph tends to contain less than $m$ nodes and therefore our landmark choice is suboptimal. If $\rho$ were smaller such that the minimal dominating set contains just $m$ nodes, then these $m$ nodes may be a better set of $m$ landmarks. To realize such an online embedding scheme requires a mechanism to both increase and decrease $\rho$ without a formidable computational cost. Extension of our algorithm to this direction is left for future work.

We theoretically showed that the time complexity of the proposed algorithm is $O(n^3)$. Numerical experiments suggested $O(n^2)$ or $O(n^3)$ time depending on the data set. If we run the LMDS every time a new data point $x_{m+r}$ arrives, with $\{x_1, \ldots, x_{m+r} \}$ as landmarks, to generate an embedding trajectory \cite{thongprayoon2023embedding}, then the total time complexity for adding the $(m+1)$th to the $(m+n)$th data points is $\sum_{r=1}^{n}O\left((m+r)^3\right) = O(n^4)$. Therefore, our online algorithm reduces the computational time by a factor of $n$. Nevertheless, faster online algorithms are desirable. Because our landmark replacement algorithm just uses the distance matrix as input, it does not have to be combined with LMDS. Among the online and out-of-sample algorithms discussed in section~\ref{sec:introduction}, L-ISOMAP, the online L-ISOMAP, and the landmark diffusion maps can be combined with
our proposed landmark replacement algorithm. This is because these three algorithms use a fixed number of landmarks. While L-ISOMAP shares the same time complexity as LMDS \cite{de2004sparse}, the online L-ISOMAP \cite{law2006incremental} and the landmark diffusion maps \cite{long2019landmark} are faster than $O(n^3)$. We used LMDS only because Theorem~\ref{thm:eigen_projection_bound}, a theorem associated with LMDS, allows us derive mathematical properties (Lemma~\ref{lemma:bound_first_part} and Theorem~\ref{thm:embedding_dislocation}) of the embedding using LMDS. Exploring online landmark replacement with the online L-ISOMAP, landmark diffusion maps, and similar algorithms may be interesting.

A potential for future work is to incorporate change-point detection into the landmark replacement framework to suppress the computation time. In the change-point detection problem, we examine the presence of a change in a given time series at each time step, usually via a statistical test \cite{chandola2009anomaly, akoglu2015graph, masuda2020guide, takeuchi2006unifying, adams2007bayesian, p2018change}. One idea is to trigger the landmark replacement algorithm only  when a new data point is detected to mark a change point. In this manner, we can reduce the number of times that one calls the actual landmark replacement routine. A potential challenge is how to maintain $L$ to be the dominating set when a new data point is an outlier to an intermediate extent such that it is not detected to be a change point whereas
it is an isolated node in the geometric graph when added. Devising landmark replacement algorithms with change-point detection is worth exploring.

\section{Acknowledgments}

We thank Lorenzo Livi for discussion. The work of Naoki Masuda was supported in part by the Air Force Office of Scientific Research (AFOSR) European Office under Grant FA9550-19-1-7024, in part by the National Science Foundation or NSF under grant DMS-2052720, in part by JSPS KAKENHI under grants JP21H04595 and 23H03414, and in part by the Japan Science and Technology Agency (JST) under grant JPMJMS2021.

\section{Supplementary material}

\subsection{Big $O$ notation}
\label{sec:bigO}

In this section, we provide the formal definition of big $O$ notation.

\begin{definition}\label{def:bigO}
Let $A$ be an unbounded subset of $\mathbb{R}^{+}$ such that $f_i:A\rightarrow\mathbb{R}$
(with $i \in \{ 1, 2 \}$) is continuous. Then, $f_1(x) = O\left(f_2(x)\right)$ as $x\rightarrow\infty$ if there exist $M > 0$ and $C > 0$ such that $\left|f_1(x)\right| \le M\left| f_2(x)\right|$, $\forall x > C$.

Alternatively, let $B$ be a subset of $\mathbb{R}$ such that $g_i: B\rightarrow\mathbb{R}$ (with $i \in \{ 1, 2 \}$) is continuous. Furthermore, let $x_0$ be a real number. Then, $g_1(x) = O(g_2(x))$ as $x\rightarrow x_0$ if there exist $\delta > 0$ and $M > 0$ such that $\left|g_1(x)\right| \le M\left| g_2(x)\right|$, $\forall x$ satisfying $0 < \left| x - x_0 \right| < \delta$.
\end{definition}

\subsection{Linear operators}
\label{sec:linear_operators}

In this section, we review basic linear operator theory that we have used in the main text.

\begin{definition}
Consider normed spaces $\left(X_1,\left\|\cdot\right\|_{X_1}\right)$ and $\left(X_2,\left\|\cdot\right\|_{X_2}\right)$, and let $T: X_1\rightarrow X_2$ be a linear map. Then, $T$ is bounded if there exists $\widetilde{M}\ge 0$ such that
\begin{equation}
\left\| T(x)\right\|_{X_2}\le\widetilde{M}\left\| x\right\|_{X_1}
\end{equation}
for all $x\in X_1$.
\label{def:bounded_operator}
\end{definition}

\begin{definition}
Consider normed spaces $\left(X_1,\left\|\cdot\right\|_{X_1}\right)$ and $\left(X_2,\left\|\cdot\right\|_{X_2}\right)$, and let $T: X_1\rightarrow X_2$ be a bounded linear map. We define the operator norm of $T$ as follows:
\begin{equation}
\left\| T\right\|_{\textnormal{op}} = \inf\{\widetilde{M}\ge 0: \left\| T(x)\right\|_{X_2}\le\widetilde{M}\left\| x\right\|_{X_1} \forall x\in X_1\}.
\end{equation}
\label{def:operator_norm}
\end{definition}
\noindent Note that Definitions~\ref{def:bounded_operator} and \ref{def:operator_norm} depend on the choice of norms.

\begin{theorem}(Einsiedler and Wart \cite{einsiedler2017functional})
Let $\left(X_1,\left\|\cdot\right\|_{X_1}\right)$ and $\left(X_2,\left\|\cdot\right\|_{X_2}\right)$ be normed spaces, and $T: X_1\rightarrow X_2$ be a bounded linear map. Then,
\begin{equation}
\left\| T(x)\right\|_{X_2}\le \left\| T\right\|_{\textnormal{op}}\left\|x\right\|_{X_1}
\end{equation}
for all $x\in X_1$.
\label{thm:bound_range_operator}
\end{theorem}

Recall that a metric space is complete if every Cauchy sequence converges to an element in the metric space, and that an inner product space $\left(X,\left\langle\cdot, \cdot\right\rangle\right)$ is a Hilbert space if it is complete with respect to its induced norm $\left\|\cdot \right\| = \sqrt{\left\langle\cdot, \cdot\right\rangle}$.

\begin{theorem} (Halmos~\cite{halmos1947finite})
Let $\left(X_1, \left\langle\cdot, \cdot\right\rangle_{X_1} \right)$ and $\left(X_2, \left\langle\cdot, \cdot\right\rangle_{X_2} \right)$ be Hilbert spaces. If $X_1$ is finite-dimensional and $T:X_1\rightarrow X_2$ is linear, then $T$ is bounded.
\label{thm:bounded_operator}
\end{theorem}

\subsection{Dependence of the accuracy of the online embedding on the number of landmarks}
\label{sec:dependence_on_m}

To assess the dependence of the quality of our online embedding algorithm on the number of landmarks, we show in Figs.~\ref{fig:stress_against_m}(a), \ref{fig:stress_against_m}(b), and \ref{fig:stress_against_m}(c) the values of $\sigma_L$ and $\sigma$ as a function of $m$ for the S-curve, SPY, and Hospital data, respectively. We set $m \in \{3, 6, 8, 10, 12, \ldots, 500 \}$ for the S-curve data and $m \in \{3, 5, 8, 11, 14, \ldots, 29, 30 \}$ for the SPY and Hospital data. We skipped relatively many $m$ values in the latter two data sets due to high computational cost. Note that $m \ge 3$ landmarks are necessary for matrix $\Delta$ to have at least two positive eigenvalues, thus enabling the embedding the data points into $\mathbb{R}^2$ using LMDS.
%

Figure~\ref{fig:stress_against_m}(a) indicates that $\sigma_L$ and $\sigma$ 
are near an acceptable range (i.e., $<0.15$)~\cite{borg2005modern} over a reasonably wide range of $m$ including $m=100$, which we used in the main analysis. With too few landmarks (i.e., $m < 50$),
$\sigma$ is substantially larger than $0.15$ and carries larger fluctuations, and $\sigma_L$ also fluctuates much.
Large values and high fluctuation of $\sigma_L$ and $\sigma$ are also present when $m\approx 220$ to $m\approx 320$, although the reasons for this behavior are unclear. Excluding these outliers, the values of $\sigma_L$ and $\sigma$ are stable for large $m$. The values of $\sigma_L$ and $\sigma$ for the SPY (Fig.~\ref{fig:stress_against_m}(b)) and Hospital (Fig.~\ref{fig:stress_against_m}(c)) data are substantially smaller than $0.15$ over the entire range of $m$ that we have examined, which includes $m=10$, the value used in our main analysis. In sum, the accuracy of the present online embedding algorithm is sufficiently robust with respect to changes in $m$ around the values used in our main analysis.

\begin{figure}[t]
    \centering
     \includegraphics[width=\textwidth]{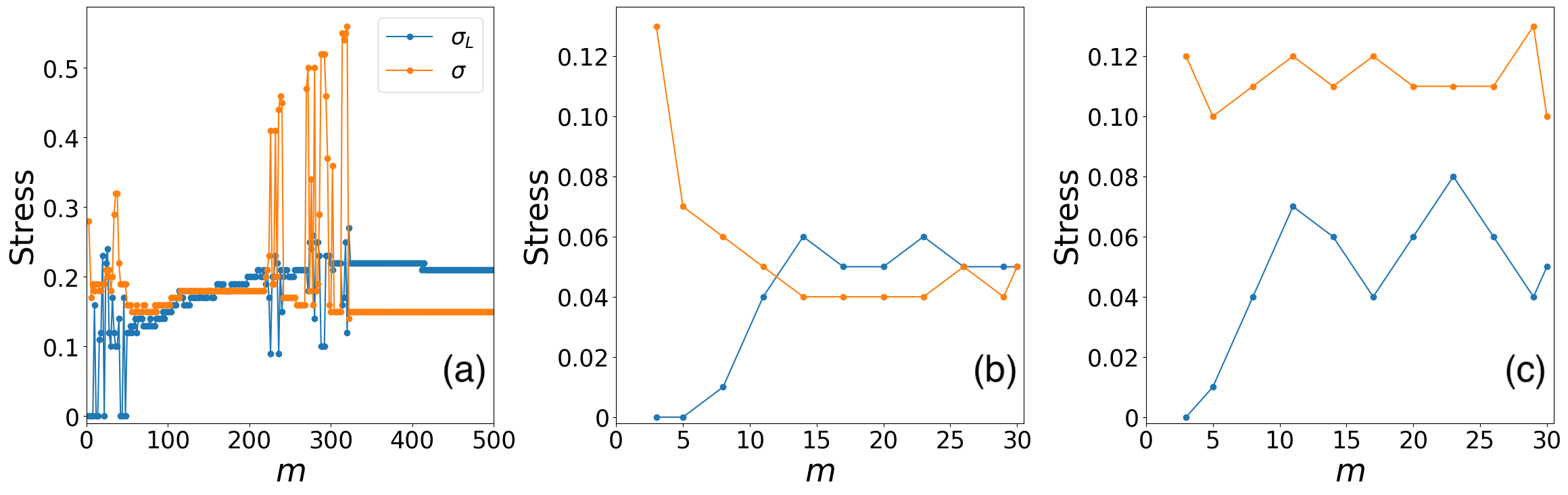}
     \caption{Stress function plotted against the number of landmarks, $m$. (a) S-curve. (b) SPY. (c) Hospital.
     }
\label{fig:stress_against_m}
\end{figure}

\subsection{Random online embedding algorithm}
\label{sec:rand_online}

Our random online embedding algorithm, which we used for comparison purposes, runs as follows.
We begin by using the first $m$ data points as the initial landmarks. Then, each new data point becomes a new landmark with probability $0.5$ upon arrival. If the new data point becomes a landmark, then we drop an existing landmark selected uniformly at random to keep the size of $L$ equal to $m$.

\subsection{Tie-decay network and its embedding}
\label{sec:tie_decay_embedding}

In this section, we briefly review tie-decay networks and its embedding.

Denote the number of nodes of the temporal network by $N$.
Consider a sequence of time-stamped contact events each of which occurs between a pair of nodes. One can represent each contact event as a tuple $(i_e, j_e, t_e)$, where $i_e$ and $j_e$ are nodes, and $t_e$ $(\ge 0)$ is the time of the contact event. For the Hospital data set, each contact event is undirected such that the order of $i_e$ and $j_e$ does not matter.
Let $B(t)=[b_{ij}(t)]$ be a time-dependent $N\times N$ weighted adjacency matrix of the tie-decay network, where $b_{ij}(t)$ represents the weight of edge $(i, j)$ at time $t\in \mathbb{R}_{\ge 0}$. For each $(i, j)$, we set
\begin{equation}
    b_{ij}(t) = \sum_{r; \tilde{t}_r \le t} e^{-\alpha(t-\tilde{t}_r)}H(t-\tilde{t}_r),
\label{eq:b_ij}
\end{equation}
where $H$ is the Heaviside step function defined by $H(x) = 1$ for $x\ge 0$ and $H(x) = 0$ for $x < 0$, $\tilde{t}_r$ is the time of the $r$th event on edge $(i, j)$, and $\alpha$ is a parameter~\cite{ahmad2021tie, sugishita2021opinion, porter2020nonlinearity+, zuo2021models}. An event on $(i, j)$ increases $b_{ij}(t)$ by $1$. We set $\alpha = 10^{-2}$ in our numerical simulations.

The exponentially decaying nature of Eq.~\eqref{eq:b_ij} is convenient for updating matrix $B$ upon any event arrival. To explain this, 
we consider a sequence of times $0\leq t_1<t_2<\cdots < t_n$, where $t_{r}$ is the $r$th time at which at least one event occurs in the entire network. We denote by $A_{r}$ the adjacency matrix of the network composed of all the events occurring at time $t_r$. The input data is then equivalent to the time series of adjacency matrices $\{ A_1, A_2, \ldots \}$. Then, the weighted adjacency matrix of the tie-decay network is given by
\begin{equation}
B(t) = \sum_{r; t_r \le t} e^{-\alpha(t - t_r)}A_r,
\label{tie_decay_adj_mat_terms}
\end{equation}
which allows a convenient updating rule:
\begin{equation}
B(t_{r+1}) = e^{-\alpha (t_{r+1} - t_{r})} B(t_{r}) + A_{r+1}.
\end{equation}

We embed $\{ B(t_1), \ldots, B(t_n) \}$ using the LMDS. This method is advantageous in that one can produce a continuous-time trajectory of a given tie-decay network, i.e., including the times at which no contact event occurs, if we use an appropriate distance measure for networks~\cite{thongprayoon2023embedding}. 
In the demonstration with the Hospital data set, we do not exploit this continuous-time nature of the trajectory because our aim is to assess the performance of our online landmark replacement algorithm. However, we use this method.

To run the algorithm for a given $\alpha$ value, we first construct the tie-decay network
from the input adjacency matrices, $\{ A_1, \ldots, A_n \}$, using Eq.~\eqref{tie_decay_adj_mat_terms}.
Second, we specify the landmarks. In the original method~\cite{thongprayoon2023embedding}, we use
$B(t_1)$, $\ldots$, $B(t_n)$ as landmarks and embed $B(t)$, where $t \notin \{t_1, \ldots, t_n \}$, using the LMDS. Here we only use the $m$ landmarks that are a subset of $\{ B(t_1), \ldots, B(t_n) \}$. Third, we run the LMDS to embed all $B(t)$'s that are not landmarks. 
For $B(t_{r})$ that are not landmarks, the usual procedures of the LMDS determine their embedding coordinates.
To embed $\{B(t) : t\ge 0 \}$, where $t\notin \{ t_1, \ldots, t_n \}$, we use the exponentially decaying nature of the tie-decay network model to obtain
\begin{equation}
B(t) = e^{-\alpha(t-t_r)}B(t_r),
\label{matrix_B}
\end{equation}
where $r$ is the unique value of $r \in \{ 1, \ldots, n\}$ that satisfies $t_r < t < t_{r+1}$. Using Eq.~\eqref{LMDS_arbitrary_coordinates}, we obtain
the embedding coordinate for $B(t)$ for any $t_r < t < t_{r+1}$ by
\begin{equation}
        \psi(B(t)) = -\frac{1}{2}L'_k(\delta_B(t)-\delta_\mu),
    \label{tie_decay_arbitrary_coordinates}
    \end{equation}
where
\begin{equation}
    \delta_B(t) = \begin{bmatrix}
    d^2(B(t),B(t_1))\\
    \vdots \\
    d^2(B(t),B(t_n))
    \end{bmatrix}.
\end{equation}

We use the unnormalized Laplacian network distance as $d$.
The Laplacian matrix $\tilde{L}(t)$ of a tie-decay network at time $t$ is given by $\tilde{L}(t)\equiv \tilde{D}(t) - B(t)$, where $\tilde{D}(t)$ is the diagonal matrix of which the $i$th diagonal entry is $\displaystyle\sum_{j=1}^N b_{ij}(t)$. The unnormalized Laplacian network distance is given by
\begin{equation}
    d_{\tilde{L}}(\tilde{L}(t_r),\tilde{L}(t_{s})) = \sqrt{\sum_{k=1}^N \left[\rho_k(\tilde{L}(t_r)) - \rho_k(\tilde{L}(t_{s})) \right]^2}\,,
\label{eq:unnormalized-Laplacian-distance-def}
\end{equation}
where  $\rho_k(P)$ is the $k$th smallest eigenvalue of symmetric matrix $P$~ \cite{masuda2019detecting, donnat2018tracking, wilson2008study}. One can replace $\sum_{k=1}^N$ by $\sum_{k=2}^{N}$ in Eq.~\eqref{eq:unnormalized-Laplacian-distance-def} because the smallest eigenvalue of a Laplacian matrix is always $0$.
Using Eq.~\eqref{matrix_B}, we can simplify Eq.~\eqref{eq:unnormalized-Laplacian-distance-def} as follows:
\begin{equation}
d_{\tilde{L}}(\tilde{L}(t),\tilde{L}(t_{s})) = \sqrt{\sum_{k=2}^N \left[e^{-\alpha(t-t_r) }\rho_k(\tilde{L}(t_r)) - \rho_k(\tilde{L}(t_{s})) \right]^2}\,.
\end{equation} 

\subsection{Growth of $\rho$\label{sec:hosp_trajectory_induced_jump}}

We plot in Figs.~\ref{fig:and_epsilon}(a), \ref{fig:and_epsilon}(b), and \ref{fig:and_epsilon}(c) the values of $\rho$ against the number of data points for the S-curve, SPY, and Hospital data, respectively. In Fig.~\ref{fig:and_epsilon}(a), $\rho$ rapidly grows till approximately the first 200 data points. In Fig.~\ref{fig:and_epsilon}(b), $\rho$ starts growing faster than before at $r \approx 1,500$ because of the nature of the SPY data; the stock prices started to explore different part of the space of the data, roughly in the beginning of year 2020 (i.e., 1,761th day). Although this interpretation is consistent with Fig.~\ref{fig:etf_embedding}(b), we further support this claim by the following computation. For simplicity, we calculated the average of the ten-dimensional vector whose each entry is the price of a stock averaged over all valid days of a year. The distance between the thus obtained vectors of two years, which we averaged over all pairs of years between 2013 and 2019, was equal to $1.80\times 10^{-4}$.
In contrast, the distance between the vector for year 2020 and that for any year between 2013 and 2019 was equal to $4.90\times 10^{-4}$ on average.
Therefore, the coordinates between 2013 and 2019 in the original space stay closer to each other than the coordinates in 2020. In Fig.~\ref{fig:and_epsilon}(c), the value of $\rho$ suddenly increases by a large amount on Monday at 2:08  PM, Monday at 2:14 PM, and Tuesday at 10:16 AM. Note that we use the logarithmic scale of $n$ in Fig.~\ref{fig:and_epsilon}(c) to uncover the first two big leaps in $\rho$ values.

\begin{figure}[t]
    \centering
     \includegraphics[width=\textwidth]{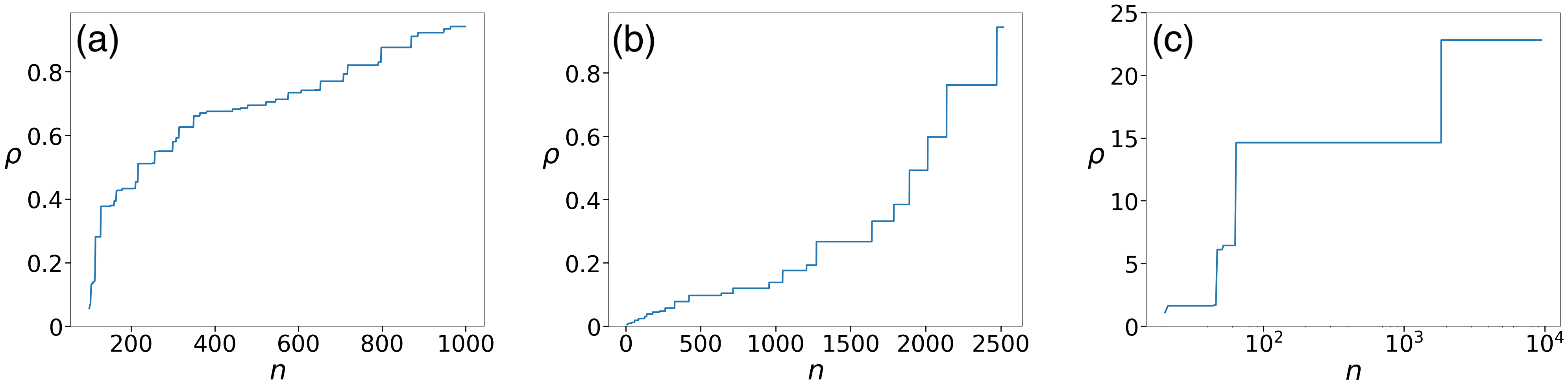}
     \caption{Evolution of $\rho$ as a function of time. (a) S-curve. (b) SPY. (c) Hospital.}
\label{fig:and_epsilon}
\end{figure}

This last phenomenon is considered to be because the embedding trajectory from Monday between 2:00 PM and 3:00 PM and Tuesday between 10:00 AM and 11:00 AM explores unexplored part of the data space compared to other segments of the same trajectory in the past. As evidence, we show in Fig.~\ref{fig:hospital_traj_spec_hour} sixteen one-hour segments of the trajectory of the Hospital temporal network, one per panel. The segments of the trajectory shown in Fig.~\ref{fig:hospital_traj_spec_hour} that include Monday 2:08 AM and 2:14 PM explore new part of the space relative to the segments in the past (i.e., the segment between Monday 1:00 PM and 2:00 PM). Similarly, apart from the segment containing Tuesday 10:16 AM, all segments up to Tuesday 10:00 AM are confined below $y = 6$ and to the right to $x = -25$. After the third sudden increment, $\rho$ $(= 22.79)$ is large enough that the existing landmarks cover the incoming networks. We recall that the $x$ and $y$ ranges shown in Fig.~\ref{fig:hosp_embedding}(b) are appoximately $x \in (-15,10)$ and $y \in (-2.5,5.5)$, respectively, which are a small subset of the embedding space shown in each panel of Fig.~\ref{fig:hospital_traj_spec_hour}. This is because we show the hourly averages of the embedding coordinate in Fig.~\ref{fig:hosp_embedding}, while we show the embedding coordinate at individual times in Fig.~\ref{fig:hospital_traj_spec_hour}; the latter is not an averaged quantity and therefore accompanies larger fluctuations.

\begin{figure}[t]
    \centering
     \includegraphics[width=\textwidth]{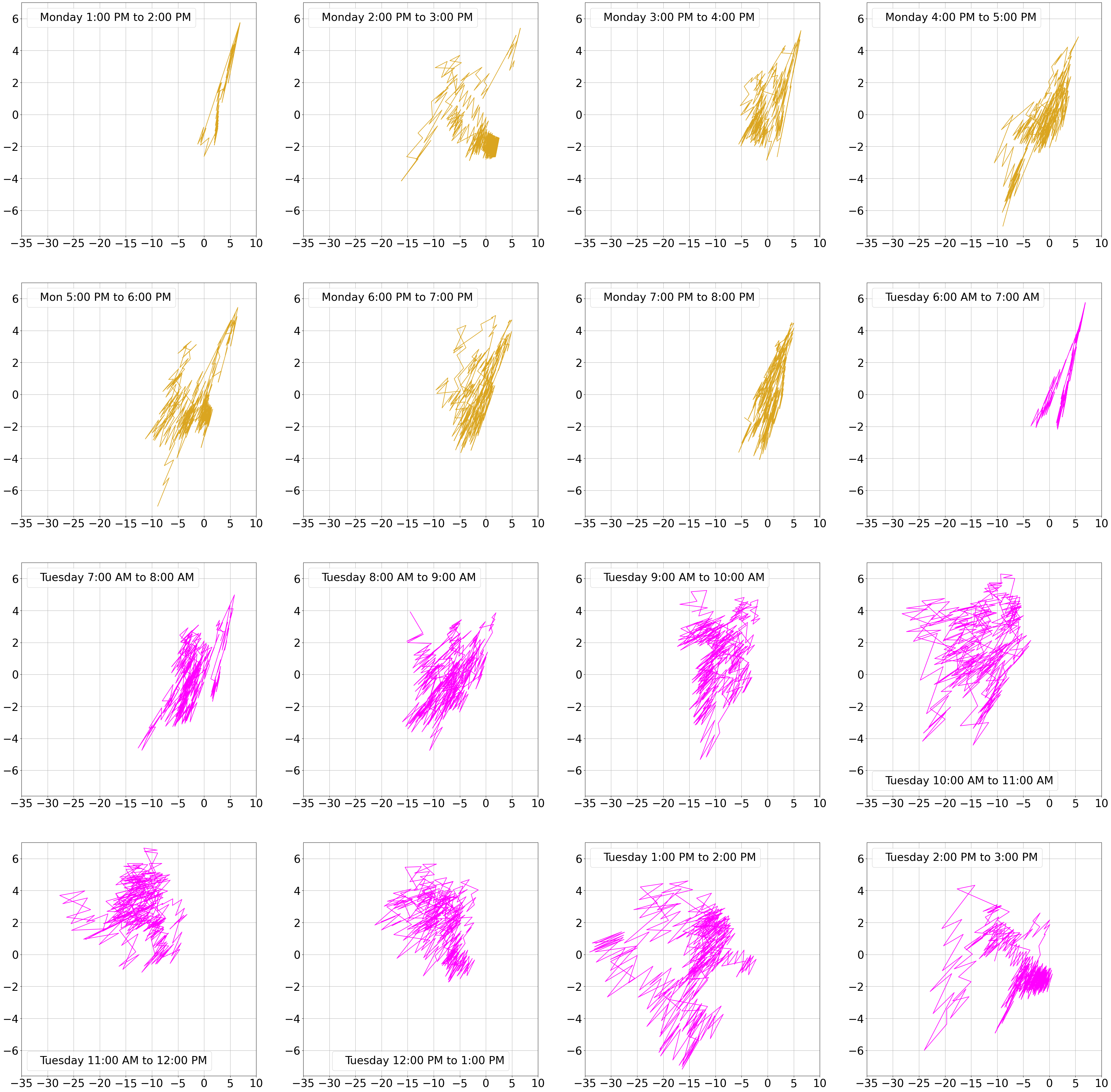}
     \caption{Two-dimensional trajectory of the Hospital temporal network in each hour from Monday 1:00 PM to 8:00 PM and from Tuesday 6:00 AM to 3:00 PM. Because 95.6\% of the contacts occur between 6:00 AM to 8:00 PM throughout the weekdays, we have omitted the trajectories outside these hours. We show the trajectory on Monday and Tuesday in different colors.}
\label{fig:hospital_traj_spec_hour}
\end{figure}

While there is no control or a stopping criterion for $\rho$, it should stop growing after its value has become sufficiently large because any newly arriving data points, including outliers, would be connected to an existing landmark when $\rho$ is sufficiently large.
In fact, Fig.~\ref{fig:and_epsilon}(c) shows that $\rho$ saturates after Tuesday at 10:16 AM (i.e., after $n = 1,838$ in the figure). Although $\rho$ does not saturate for the other two data sets
(see Figs.~\ref{fig:and_epsilon}(a) and \ref{fig:and_epsilon}(b)), we expect that it roughly saturates if the data are longer in the case of the SPY data. In the case of the S-curve, $\rho$ would continue to increase if the S-like shape further continues. However, this situation is obviously artificial.

\clearpage
\clearpage
\bibliographystyle{chronological}
\bibliography{bibliography.bib}

\end{document}